\documentclass[
  aps,pre,
  preprint,           
  nofootinbib,
  superscriptaddress,
  longbibliography
]{revtex4-2}

\usepackage[T1]{fontenc}
\usepackage[utf8]{inputenc}
\usepackage{microtype}
\usepackage{bm}
\usepackage{siunitx}
\usepackage{booktabs}
\usepackage{tabularx}
\makeatletter
\@ifpackageloaded{inputenc}{%
  \DeclareUnicodeCharacter{2009}{\,}%
}{}
\makeatother
\usepackage[section]{placeins}

\usepackage{mathtools}           
\usepackage{amsthm}
\usepackage{newtxtext}
\usepackage[amsthm]{newtxmath}
\mathtoolsset{showonlyrefs}

\usepackage{graphicx}
\graphicspath{{fig/}{figures/}}
\usepackage[most]{tcolorbox}
\usepackage{enumitem}
\DeclareSIUnit\bit{bit}
\DeclareSIUnit\nat{nat}
\usepackage{empheq}
\usepackage[caption=false]{subfig}

\usepackage{nameref}
\usepackage[colorlinks=true,
  linkcolor=blue,
  citecolor=blue,
  urlcolor=blue]{hyperref}

\theoremstyle{plain}
\newtheorem{theorem}{Theorem}

\newtheorem{lemma}{Lemma}
\newtheorem{corollary}{Corollary}
\theoremstyle{definition}
\newtheorem{definition}{Definition}
\newtheorem{remark}{Remark}

\begin{document}

\title{A Single--Index Theory of Optimal Branching: Murray Laws, Gilbert Networks, and Young--Herring Junctions}

\author{Justin Bennett}
\email{jmb15@illinois.edu}
\affiliation{Department of Physics, University of Illinois Urbana--Champaign, Urbana, Illinois 61801, USA}
\thanks{Independent work; not conducted under, nor funded by, the University of Illinois research program.}

\date{\today}

\begin{abstract}
Murray-type flux--radius laws, Gilbert-type concave flux costs, and Young--Herring triple-junction angle balances are usually treated as distinct pieces of theory. This paper shows that, within a natural class of quadratic scale-free ledgers, all three arise from a single underlying structure controlled by one dimensionless index $\chi$. A branched network is represented as a graph whose edges carry a flux $Q$, a radius $r$, and a per-length ledger $\mathcal P(Q,r)$ that aggregates transport dissipation and structural burden. Under locality, evenness in $Q$, linear-response quadratic dependence, and an exact scale-free homogeneity ansatz in $(Q,r)$ on an open scaling cone, any admissible ledger in that regime reduces to a two-term power law in $r$ multiplying $Q^2$ and a purely structural term. For networks that are locally optimal in the radii and stationary under translations of interior nodes, this quadratic scale-free ledger simultaneously implies: (i) a flux--radius law $|Q|\propto r^\alpha$ with generalized Murray closures at degree-3 nodes; (ii) a Young--Herring-type vector balance with weights $r^m$ and a fixed symmetric Y-junction angle; and (iii) an effective flux-only cost $\mathcal P^\ast(Q)\propto |Q|^\beta$ of Gilbert/branched-transport type. The exponents $\alpha$ and $\beta$, the symmetric junction angle, and the relative ledger weights of transport and structure are all determined by a single index $\chi = m/(m+p)$. A rigidity theorem shows conversely that, within the quadratic ledger class and under matching homogeneity assumptions on the branchwise minimizer and minimized flux-only cost, any ledger that produces power-law $r^\ast(Q)$ and $\mathcal P^\ast(Q)$ must belong to this two-term family and hence obey the same Murray--Gilbert--Young dictionary. Examples for Poiseuille, diffusive, and geophysical trees illustrate how $\chi$ can be inferred from geometry and used as a falsifiable order parameter for scale-free branching architectures.
\end{abstract}

\maketitle
\newpage
\tableofcontents

\section{Introduction}
\label{sec:intro}
Fluxes in nature rarely move through homogeneous bulk; they are channeled through branching networks. Blood is delivered by vascular trees, water drains through river basins, electrical currents are routed through power grids, samples and reagents are distributed through microfluidic chips, and metabolites or gases are conveyed by tissue vasculature. In every case, geometry is constrained: wide conduits carry flux easily but are costly to build and maintain, whereas narrow conduits are cheap but dissipative. The working architecture of such systems reflects a compromise between transport performance and a structural or informational cost.

Across physics, geophysics, and biology, this compromise repeatedly appears in three closely related mathematical structures. Murray-type flux–radius relations tie flux $Q$ to radius $r$ along branches through a power law $|Q|\propto r^\alpha$ and enforce additive closures such as $r_0^\alpha = r_1^\alpha + r_2^\alpha$ at degree-3 junctions, as in classical vascular optimality, plant hydraulics, and microfluidic network design \cite{Murray1926,Murray1926b,ShermanJGP,Uylings1977,ZamirOptimalityTrees,McCullohYarrum,EmersonGeneralisedMurray}. Concave flux costs of branched-transport or optimal-network type,
\[
  \sum_e \ell_e |Q_e|^\beta,
  \qquad 0<\beta<1,
\]
encode the energetic advantage of sharing trunks rather than routing fluxes along independent paths \cite{Gilbert1967,XiaBranchedTransport,BernotCasellesMorel,BouchitteButtazzo,AmbrosioFuscoPallara}. Young–Herring triple-junction balances and Steiner-tree angle laws relate branch radii to opening angles through effective line or surface tensions and yield Snell-like refraction rules for optimal junction geometry \cite{Herring1951,CahnHilliard1958,MorganGMT}.

In much of the existing literature these structures are introduced separately. Murray-type laws arise from balancing viscous dissipation against vessel volume or surface area; concave flux costs emerge from coarse-graining microscopic transport into an effective network functional; Young–Herring conditions are derived from interfacial thermodynamics or geometric optimization of curves and surfaces. Taken together, however, they suggest a more unified picture and motivate a structural question:

\medskip
\emph{Under what minimal assumptions on a local branch ledger are Murray laws, concave flux costs, and Young–Herring node balances different projections of the same underlying object, and how many independent parameters remain once scale-free behavior is imposed?}
\medskip

Previous work introduced an entropy-per-information-cost (EPIC) variational principle for maintained transport networks \cite{BennettEPIC}. There, transport dissipation and structural upkeep were expressed in a common currency (J/bit) and a node-level ledger
\[
  L(r) = P_{\rm flow}(r) + \Lambda B(r)
\]
was analyzed for steady Poiseuille flow at a single Y-junction, with $P_{\rm flow}(r)$ the pumping power, $B(r)$ a structural bit rate scaling as a power of $r$, and $\Lambda>0$ a priced-bit multiplier. Stationarity of this EPIC ledger yielded a generalized Murray scaling $Q\propto r^\alpha$, a Young–Herring-type vector balance fixing bifurcation angles, a universal power partition between flow and structure at optimality, and, after eliminating radii, a strictly concave flux-only cost $\propto |Q|^\gamma$ with a routing index governing Snell-like refraction in heterogeneous price fields. A preregistered test on retinal bifurcations supported this picture. Structurally, a ledger built from two competing powers of $r$ was shown to organize flux–radius laws, node angles, concave flux costs, and routing behavior as a single bundle.

The present work abstracts and completes that structural picture. System-specific thermodynamics, information-theoretic units, and Poiseuille details are replaced by an abstract, dimensionless ledger, and the retinal EPIC model becomes a single example inside a broader class of quadratic, scale-free ledgers. Within this class, Murray-type flux–radius relations, Gilbert-type concave flux costs, and Young–Herring junction balances turn out to be different projections of a single underlying structure controlled by one dimensionless index.

The starting point is a branched network modeled as a graph whose edges $e$ carry a scalar flux $Q_e$, an effective radius $r_e$, and a per-length ledger $\mathcal P(Q_e,r_e)$ interpreted as a local ``price density'' combining transport dissipation with a structural or informational burden. In a putative scale-free regime, the ledger is required to satisfy four structural conditions, made precise in Sec.~\ref{sec:branch-model}: locality and additivity along branches (so that a homogeneous edge of length $\ell$ contributes $\ell\,\mathcal P(Q,r)$); evenness in flux, $\mathcal P(-Q,r)=\mathcal P(Q,r)$; a quadratic dependence on $Q$ at fixed $r$, reflecting linear-response entropy production and allowing a decomposition $\mathcal P(Q,r)=A(r)Q^2+B(r)$ with $A,B>0$; and an exact scale-free homogeneity ansatz, namely that a joint rescaling $(Q,r)\mapsto(\lambda^u Q,\lambda r)$ rescales the ledger as
\[
  \mathcal P(\lambda^u Q,\lambda r)=\lambda^v \mathcal P(Q,r)
\]
on an open scaling cone. The last condition is decisive: it is not a generic consequence of asymptotic power-law behavior, but a strong requirement that $\mathcal P$ itself be homogeneous of degree $v$ in $(Q,r)$ on a scale-free window. Under this hypothesis, an Euler-type argument applied to $A$ and $B$ shows that any such ledger reduces, on that cone, to a two-term power law
\begin{equation}
  \mathcal P(Q,r)
  =
  a\,\frac{Q^2}{r^{p}}
  +
  b\,r^{m},
  \qquad a,b>0,\; p>0,\; m>0,
  \label{eq:intro-two-term}
\end{equation}
used throughout the paper. This two-term form is best viewed as a homogeneous normal form for an intermediate asymptotic regime in the sense of Barenblatt \cite{BarenblattScaling}: outside that window, additional scales, logarithmic corrections, curvature, anisotropy, or nonlocal effects may generate deviations from the pure power laws derived below.

Once the ledger takes the form \eqref{eq:intro-two-term}, local optimality of the associated network functional
\[
  \mathcal F[\{Q_e,r_e\}] = \sum_e \ell_e\,\mathcal P(Q_e,r_e)
\]
produces three families of conditions aligned with the classical structures above. Minimizing $\mathcal P(Q,r)$ with respect to $r$ at fixed $Q$ yields a unique branchwise optimum $r^\ast(Q)$ with $|Q|\propto r^\alpha$ and exponent $\alpha=(m+p)/2$, and flux conservation at degree-3 junctions with one parent and two daughters enforces generalized Murray closures $r_0^\alpha = r_1^\alpha + r_2^\alpha$ (Sec.~\ref{sec:murray}). Stationarity of $\mathcal F$ under infinitesimal translations of interior nodes, with fluxes and radii held at their branchwise optima, leads to a Young–Herring triple-junction balance
\[
  \sum_{e\in E(v)} r_e^m\,\hat{\bm t}_e = \bm 0,
\]
with unit tangents $\hat{\bm t}_e$ at the node and effective ``skeleton tensions'' $\propto r^m$; for symmetric Y-junctions this yields an angle law
\[
  \cos\frac{\theta}{2} = 2^{(m-p)/(m+p)}
\]
(Sec.~\ref{sec:angles}). Eliminating $r$ using the branchwise optimum produces an effective flux-only ledger $\mathcal P^\ast(Q) \propto |Q|^\beta$ with concavity exponent $\beta = 2m/(m+p)$, so that
\[
  \mathcal F^\ast[\{Q_e\}]
  \propto
  \sum_e \ell_e |Q_e|^\beta,
\]
and the node balance can equivalently be written as $\sum_{e\in E(v)} |Q_e|^\beta\,\hat{\bm t}_e = \bm 0$ (Sec.~\ref{sec:concavity}). For $0<m<p$ one has $0<\beta<1$, and the effective flux-only cost coincides with the classical Gilbert/branched-transport functional \cite{Gilbert1967,BernotCasellesMorel,XiaBranchedTransport}. Thus the classical Murray law, Young–Herring balance and angle law, and concave flux-only functional all arise as consequences of the same two-term ledger \eqref{eq:intro-two-term}.

Within this ledger class, the a priori independent quantities
\[
  \alpha \quad \text{(flux–radius exponent)},\qquad
  \beta \quad \text{(concavity exponent)},\qquad
  \theta \quad \text{(symmetric Y-junction angle)}
\]
are not free. All are controlled by a single dimensionless index
\begin{equation}
  \chi := \frac{m}{m+p}\in(0,1),
\end{equation}
introduced in Sec.~\ref{sec:chi-dictionary}. Eliminating $(m,p)$ in favour of $(\alpha,\beta,\theta)$ gives the “Snell–Murray dictionary”
\begin{equation}
  \chi
  =
  1 - \frac{p}{2\alpha}
  =
  \frac{\beta}{2}
  =
  \frac{1}{2}\Big(1 + \log_2\cos\tfrac{\theta}{2}\Big),
  \label{eq:intro-dictionary}
\end{equation}
together with equivalent relations such as $\cos(\theta/2)=2^{\beta-1}$ and $\alpha=p/(2-\beta)$. Once the transport exponent $p$ is set by the underlying linear transport law, the triple $(\alpha,\beta,\theta)$ lies on a one-parameter curve labeled by $\chi$, and any two of $\{\alpha,\beta,\theta\}$ determine the third. In the EPIC thermodynamic interpretation \cite{BennettEPIC}, the same index $\chi$ also equals the fraction of the local ledger spent on transport rather than structure at the branchwise optimum; this ledger-sector meaning is revisited in Sec.~\ref{sec:chi-dictionary}.

The direction of implication can also be reversed. Starting from a general quadratic ledger
\[
  \mathcal P(Q,r) = A(r)Q^2 + B(r), \qquad A(r),B(r)>0,
\]
consider a regime in which the branchwise minimizer $r^\ast(Q)$ is a pure power of $Q$ and the minimized flux-only ledger $\mathcal P^\ast(Q)$ is a pure power of $|Q|$. Under these homogeneity assumptions alone, the analysis in Sec.~\ref{sec:rigidity} shows that $A(r)$ and $B(r)$ must be monomials $A(r)\propto r^{-p}$ and $B(r)\propto r^m$, with exponents related by the dictionary \eqref{eq:intro-dictionary}. Once this is established, the radius-weighted Young–Herring balance with weights $r^m$ and its flux-only counterpart follow from the forward derivation in Secs.~\ref{sec:murray}–\ref{sec:concavity}; they then become consequences of the rigidity result rather than additional hypotheses. Within the quadratic ledger class, subject to the exact homogeneity hypotheses stated in Theorem~\ref{prop:rigidity-quadratic}, there is therefore no alternative way to obtain simultaneous power-law flux–radius scaling and power-law flux-only concavity (and hence the associated Young–Herring node balances) other than to be in the two-term ledger class \eqref{eq:intro-two-term}. In this precise sense, Murray-type laws, Gilbert concavity, and Young–Herring node balances are different faces of a single underlying structure once the ledger is constrained by these homogeneity assumptions.

The main results can be summarized as follows:
\begin{itemize}
  \item \textbf{Two-term ledger from homogeneity.} Under locality, evenness in $Q$, quadratic linear-response structure, and the scale-free homogeneity ansatz of Definition~\ref{def:admissible-ledger}, the ledger on any exactly scale-free cone can be written, by an Euler-type argument, in the two-term power-law form \eqref{eq:intro-two-term} (Sec.~\ref{sec:branch-model}).

  \item \textbf{Unification and dictionary.} For such ledgers, local optimality forces Murray-type flux–radius laws, Young–Herring junction balances and angle laws, and a concave Gilbert-type flux functional whose exponents and angles are tied together by a single index $\chi$ via the dictionary \eqref{eq:intro-dictionary} (Secs.~\ref{sec:murray}–\ref{sec:chi-dictionary}).

  \item \textbf{Rigidity.} Conversely, within the quadratic ledger class, any ledger for which the branchwise minimizer $r^\ast(Q)$ and the minimized flux-only cost $\mathcal P^\ast(Q)$ are exact power laws of $Q$ on a scale-free cone must belong to the two-term class \eqref{eq:intro-two-term}, with exponents constrained by the same dictionary (Sec.~\ref{sec:rigidity}); once this holds, the power-law Young–Herring weights follow automatically.
\end{itemize}

The intention is not to claim that all branched systems obey these axioms at all scales, but rather to identify a natural quadratic, scale-free class and to describe the geometric and energetic consequences that follow whenever a network flows into this class in an intermediate asymptotic regime.

The remainder of the paper is organized as follows. Section~\ref{sec:background} reviews several families of branched systems and explains how the ledger axioms arise in each. Section~\ref{sec:branch-model} introduces the branch-network model and the admissible ledger class, and proves the homogeneity lemma that reduces admissible ledgers to the two-term form \eqref{eq:intro-two-term}. Sections~\ref{sec:murray}–\ref{sec:concavity} derive the flux–radius law, the generalized Murray closure, the Young–Herring node balance, symmetric Y-junction angles, and the effective flux-only concavity exponent. Section~\ref{sec:chi-dictionary} introduces the index $\chi$, establishes the Snell–Murray dictionary linking $(\alpha,\beta,\theta)$, and shows that $\chi$ measures the ledger-sector split between transport and structure. Section~\ref{sec:rigidity} states and proves the rigidity theorem for the general quadratic ledger class. Section~\ref{sec:applications} illustrates the framework on Poiseuille-type trees, diffusive and Ohmic conduction networks, and coarse geophysical models. Section~\ref{sec:discussion} relates these results back to the original EPIC formulation and discusses how $\chi$ can serve as a falsifiable order parameter for branched networks in experiments and simulations.

\section{Motivating branched systems}
\label{sec:background}
Before turning to the abstract ledger of Sec.~\ref{sec:branch-model}, it is useful to recall several classes of systems in which Murray-type relations, concave transport costs, and junction-angle rules have all appeared, but typically as separate pieces of theory. The aim here is not an exhaustive review, but a compact justification for two modeling ingredients: a quadratic dependence on flux $Q$ (reflecting linear-response entropy production) and power-law dependence on radius $r$ (reflecting geometric and upkeep scaling within an approximately scale-free window).

\subsection{Microfluidic and Poiseuille-type distribution trees}
\label{subsec:microfluidics}
Lab-on-a-chip devices, microreactors, and compact heat exchangers often rely on branching microfluidic trees to distribute or collect fluid between a small number of inlets and many outlets \cite{SquiresQuakeReview,BruusMicrofluidics,KirbyMicroNano,KnightMicrofluidicPolymer}. In many designs the flow is steady, laminar, and well described by the Hagen--Poiseuille law in channels of approximately circular or rectangular cross-section. For a straight cylindrical segment of radius $r$ and length $\ell$, the volumetric flux obeys
\[
  Q = \frac{\pi r^{4}}{8\mu \ell}\,\Delta P,
\]
so that, at fixed $Q$, the viscous dissipation power per unit length scales as $Q^{2} r^{-4}$. This identifies the transport exponent $p=4$ in the two-term ledger and exemplifies the generic quadratic-in-$Q$ structure associated with linear irreversible thermodynamics; Sec.~\ref{subsec:EPIC-principle} shows how this $Q^{2}/r^{p}$ form arises more generally from linear transport and slender-branch geometry.

Structural costs in such devices are naturally tied to channel size. Fabrication, footprint, and fouling penalties typically grow with cross-sectional measures such as perimeter ($\propto r$) or area ($\propto r^{2}$), depending on whether wall area or occupied volume dominates the burden. To leading order, this is captured by a power law $r^{m}$ with $m=1$ or $m=2$. Over several decades of scale, many microfluidic trees and Poiseuille-type vascular networks exhibit approximately constant exponents in flux--radius relations and reuse similar branching motifs, indicating a practical scale-free window. Within that window, a ledger of the form
\[
  \mathcal P(Q,r) \sim a\,\frac{Q^2}{r^{4}} + b\,r^{m}
\]
is therefore a natural coarse-grained summary of transport dissipation and structural upkeep.

\subsection{Geophysical transport networks}
\label{subsec:geophysical}
River basins, delta networks, and subsurface drainage systems provide a second class of branched transport structures \cite{RodriguezIturbeRinaldo,RinaldoMinimumEnergy}. At catchment scales, empirical geomorphological laws relate contributing area, channel width, and mean discharge; at finer scales, the cascade of tributaries feeding larger streams qualitatively resembles an optimal-transport tree. Continuum and graph-based models often represent the energetic cost per unit length of a channel as the sum of a transport term and a structural term. In regimes where an effective linear relation between driving potential and flux holds, hydraulic losses scale as $Q^{2}$ divided by an effective cross-sectional measure, leading again to a ledger contribution $\propto Q^{2}/r^{p}$ with an exponent $p$ determined by the channel codimension and roughness. Strongly turbulent regimes may modify the effective exponent, but over intermediate ranges of discharge and scale, a quadratic approximation is commonly employed.

Structural costs in geophysical channels arise from excavation, bank stabilization, and sediment transport capacity, and typically scale with cross-sectional area, wetted perimeter, or erodible volume \cite{BanavarEfficientNetworks,BejanConstructalReview}. These contributions are well modeled, at leading order, by a power law $r^{m}$ in an effective radius. Concave flux costs of the form $\sum_{e}\ell_{e}|Q_{e}|^{\beta}$ with $0<\beta<1$ are introduced phenomenologically to capture the observed consolidation of flow into a few large rivers rather than many small ones \cite{BanavarEfficientNetworks,BejanConstructalReview,BejanConstructalPattern,BohnMagnascoPRL,KatiforiPRL,CorsonPRL}. The EPIC ledger \eqref{eq:intro-two-term} provides a microscopic interpretation of such concavity within a quadratic, scale-free regime: the concavity exponent $\beta$ emerges from $(m,p)$, and the same exponents control flux--width relations and junction-angle statistics when a scale-free window is present.

\subsection{Biological tissue networks}
\label{subsec:tissues}
Biological transport trees range from micron-scale capillary beds and leaf venation to centimetre-scale pulmonary airways and coronary arteries. Many such systems have long been analyzed with Murray-type optimality arguments. In the original vascular setting, viscous pumping power $\propto Q^{2}/r^{4}$ was balanced against a volume-based maintenance cost $\propto r^{2}$, yielding the cubic relation $r_{0}^{3}=r_{1}^{3}+r_{2}^{3}$ for blood vessels \cite{Murray1926,Murray1926b}. Subsequent work extended this framework to surface-based upkeep ($m=1$), non-Newtonian rheology, and whole-organ constraints, and applied it to xylem, phloem, and respiratory trees, including systematic treatments of optimal diameters and branching angles \cite{ShermanJGP,Uylings1977,ZamirOptimalityTrees,EmersonGeneralisedMurray,PriesSecombReview,WestBrownEnquistScaling}. Parallel literatures introduced concave flux costs and angle-selection principles at the organ or organism level, often motivated by efficient network architecture or uniform stress arguments.

In these tissue trees, laminar Newtonian hydraulics again identifies a transport term $\propto Q^{2}/r^{4}$, while upkeep, metabolic maintenance, or material constraints contribute geometric powers of $r$. Radii and fluxes typically span several decades, and statistical exponents for Murray-type laws and angle relations remain approximately constant across scales, indicating that many tissues operate within an approximate scale-free regime. In such a regime, an abstract ledger of the form $aQ^{2}/r^{p}+b r^{m}$ with $p=4$ and $m$ reflecting volume- or surface-based pricing captures the leading competition between transport dissipation and structural burden.

\medskip

These three families of examples point to the two modeling choices formalized in Sec.~\ref{sec:branch-model}. First, the quadratic dependence on $Q$ can be treated as a generic consequence of linear, local transport plus slender-branch geometry; Sec.~\ref{subsec:EPIC-principle} derives the $Q^{2}/r^{p}$ scaling and identifies $p$ with an effective cross-sectional exponent, independent of the specific medium. Second, the two-term ledger \eqref{eq:intro-two-term} is best viewed as the leading homogeneous approximation to a more complicated microscopic ledger in an intermediate, approximately self-similar window. The rest of the paper simply works out the geometric and energetic consequences that follow whenever a branched system spends part of its scale range in such a quadratic, scale-free regime.

\section{Homogeneity I: branch-network model and axiomatic ledgers}
\label{sec:branch-model}

The discussion now turns from physical examples to an abstract model of branched networks and to the precise assumptions placed on the local ledger. The description is deliberately general, and all subsequent rigidity statements are \emph{conditional} on these axioms and on working, for part of the analysis, in an idealised exactly scale-free window, interpreted as a leading-order approximation to an intermediate asymptotic regime.

\subsection{Graph model and flux conservation}
\label{subsec:graph-model}

A branched network is represented as a finite graph $G=(V,E)$ smoothly embedded in $\mathbb R^{d}$ with $d\geq 2$. The formalism allows general graphs with loops, but in the Murray-law and angle derivations of Secs.~\ref{sec:murray}--\ref{sec:concavity} the analysis is restricted to locally tree-like regions without cycles. Each edge $e\in E$ is modeled as a slender branch characterized by
\begin{itemize}
  \item an arclength $\ell_e>0$,
  \item an effective radius (or half-width) $r_e>0$,
  \item a scalar flux $Q_e\in\mathbb R$.
\end{itemize}
The ledger depends only on the magnitude $|Q_e|$; the sign of $Q_e$ is encoded by assigning an arbitrary orientation to each edge. Throughout, $Q_e\in\mathbb R$ denotes an \emph{oriented} flux and $|Q_e|$ its magnitude. When scalar relations such as $Q_0 = Q_1 + Q_2$ are written at a degree-3 ``Murray node'', the local edge orientations are always chosen so that $Q_i>0$ at that node, and the $Q_i$ can be read as flux magnitudes there.

At interior nodes the network is assumed to conserve flux. Let $v\in V$ denote an interior node and $E(v)$ the set of edges incident on $v$. For each $e\in E(v)$ let $\sigma_{ve}=\pm 1$ indicate whether the chosen orientation of $e$ is outgoing or incoming at $v$. Flux conservation then reads
\begin{equation}
  \sum_{e\ni v} \sigma_{ve} Q_e = 0
  \quad\text{for each interior node } v\in V.
\end{equation}
Nodes at which net flux is prescribed play the role of sources and sinks and are not varied in the optimization problems considered later.

Throughout, a dimensionless parametrization of the ledger is used, but it is useful to keep the underlying dimensionful quantities in view. In a given physical system the per-length ledger has the schematic form
\[
  \mathcal P_{\rm phys}(Q_{\rm phys},r_{\rm phys})
  =
  a_0\,\frac{Q_{\rm phys}^2}{r_{\rm phys}^{p}}
  +
  b_0\,r_{\rm phys}^{m},
  \qquad a_0,b_0>0,
\]
where $Q_{\rm phys}$ and $r_{\rm phys}$ carry their usual units (for example m$^3$/s and m for Poiseuille flow), $p$ and $m$ are dimensionless exponents, and $a_0,b_0$ absorb material parameters and any entropy- or information-pricing factors (such as $T_{\rm env}$ or an energy-per-bit scale). Choosing reference values $Q_{\rm ref},r_{\rm ref}>0$ and a reference ledger scale $P_{\rm ref}>0$, define
\[
  Q := \frac{Q_{\rm phys}}{Q_{\rm ref}},
  \qquad
  r := \frac{r_{\rm phys}}{r_{\rm ref}},
  \qquad
  \mathcal P(Q,r) := \frac{\mathcal P_{\rm phys}(Q_{\rm phys},r_{\rm phys})}{P_{\rm ref}}.
\]
With the choice $P_{\rm ref}=a_0 Q_{\rm ref}^2/r_{\rm ref}^{p}$ the dimensionless ledger takes the two-term form
\[
  \mathcal P(Q,r)
  =
  a\,\frac{Q^2}{r^{p}}
  +
  b\,r^{m},
\]
with $a,b>0$ dimensionless combinations of $(a_0,b_0,Q_{\rm ref},r_{\rm ref})$. All scale-free relations derived below depend only on the exponents $(m,p)$, or equivalently on the index $\chi=m/(m+p)$, and are insensitive to the particular choice of reference scales or to the overall multiplicative normalization of $\mathcal P$. From this point on the subscript ``phys'' is suppressed and $Q$, $r$, and $\mathcal P$ denote these dimensionless variables. In the unit-bearing EPIC interpretation, the logical order is that a \emph{physical} ledger is first fixed, in which the transport and structural contributions have already been expressed in a common unit (for example J/bit per unit length), and only then is the passage made to the dimensionless ledger and to the parameter~$\chi$. Whenever $\chi$ is interpreted as ``the fraction of the ledger spent on transport,'' this is always with respect to such a prior choice of common units.

\subsection{Admissible ledgers and scale-free characterization}
\label{subsec:admissible-ledger}
The central modeling ingredient is the local ledger $\mathcal P(Q,r)$ assigned per unit arclength of branch. This subsection formalizes the class of ledgers under consideration and records the consequences of the scale-free homogeneity assumption. In particular, it shows that, in any regime where a strong homogeneity relation holds exactly on an open cone in $(Q,r)$, a quadratic (linear-response) ledger can, on that cone, be parametrized by two exponents $(m,p)$ appearing as simple monomials of $r$.

\begin{definition}[Admissible branch ledger]
\label{def:admissible-ledger}
A per-length branch ledger $\mathcal P(Q,r)$ is called \emph{admissible} if it satisfies:
\begin{enumerate}[label=(\alph*),nosep]
  \item \textbf{Locality and additivity along branches:} the total cost of a branch with arclength $\ell$ is
        \[
          \mathcal C_{\text{branch}} = \int_0^\ell \mathcal P(Q(s),r(s))\,ds.
        \]
        For a homogeneous branch with constant $Q$ and $r$ one has $\mathcal C_{\text{branch}} = \ell\,\mathcal P(Q,r)$. When the ledger later takes the two-term form~\eqref{eq:ledger-def} and the flux is constant along the branch, $Q(s)\equiv Q$, any local minimizer of $\mathcal C_{\text{branch}}$ with respect to radius profiles $r(s)>0$ at fixed $\ell$ and $Q$ has a constant radius $r(s)\equiv r^\ast(Q)$: the integrand depends on $s$ only through $r(s)$ and not on derivatives $r'(s)$, so the Euler--Lagrange equation reduces pointwise to $\partial_r\mathcal P(Q,r(s))=0$, and for the two-term ledger the map $r\mapsto\mathcal P_{\rm br}(Q,r)$ has a unique minimizer $r^\ast(Q)$ for each fixed $Q$ (Lemma~\ref{lem:flux-radius}), which forces $r(s)\equiv r^\ast(Q)$ almost everywhere.

  \item \textbf{Evenness in flux:}
        \[
          \mathcal P(-Q,r) = \mathcal P(Q,r)
          \quad\text{for all } Q,r>0,
        \]
        so the ledger depends on $Q$ only through $Q^2$.

  \item \textbf{Quadratic transport sector (linear response):} there
        exist scalar functions $A,B:(0,\infty)\to(0,\infty)$ with $A,B\in C^1((0,\infty))$ such that
        \begin{equation}
          \label{eq:quadratic-decomposition}
          \mathcal P(Q,r)
          =
          A(r)\,Q^2 + B(r),
        \end{equation}
        for all $Q\in\mathbb R$ and $r>0$. This encodes the standard linear-response form for local entropy production at fixed geometry \cite{Onsager1931,DeGrootMazur,KondepudiPrigogine}.

  \item \textbf{Scale-free homogeneity:} there exist exponents
        $u,v\in\mathbb R$ and a nonempty open set $D\subset(0,\infty)\times(0,\infty)$ (the \emph{scale-free regime}) with the following properties:
        \begin{itemize}[nosep]
          \item[(i)] $D$ is \emph{locally scale-invariant}: whenever $(Q,r)\in D$, there exists $\varepsilon>0$ such that $(\lambda^u Q,\lambda r)\in D$ for all $\lambda\in(1-\varepsilon,1+\varepsilon)$;
          \item[(ii)] the projection of $D$ onto the $r$-axis has at least one connected component that is a nonempty open interval in $(0,\infty)$. One such component is fixed and denoted by $I$. For each $r\in I$, the slice $J_r = \{Q>0 : (Q,r)\in D\}$ contains a nontrivial open interval in $Q$.
        \end{itemize}
        On $D$ the ledger is exactly homogeneous of degree $v$ under the joint rescaling $(Q,r)\mapsto(\lambda^u Q,\lambda r)$:
        \begin{equation}
          \label{eq:scale-free-homogeneity}
          \mathcal P(\lambda^u Q,\lambda r)
          =
          \lambda^v\,\mathcal P(Q,r)
        \end{equation}
        whenever $(Q,r)$ and $(\lambda^u Q,\lambda r)$ both lie in $D$.
\end{enumerate}
\end{definition}

Throughout, the term \emph{scale-free regime} refers to such an open cone $D$ together with its chosen projection interval $I$. When the term ``cone of fluxes'' in one dimension is used (for example in Theorem~\ref{prop:rigidity-quadratic}), it refers to an interval of positive flux values that is locally invariant under multiplication by $\lambda$ near $1$ (a scale-invariant open ray).

Assumptions (a)–(c) are natural for any local energetic accounting in a one-dimensional branch and encode linear irreversible thermodynamics at fixed geometry. Assumption (d) is a strong homogeneity ansatz: there is presumed to exist a window of branch scales and fluxes, far from microscopic cutoffs and from global system size, in which the coarse-grained ledger behaves as a homogeneous function of $(Q,r)$ in the sense of~\eqref{eq:scale-free-homogeneity}, in line with Barenblatt’s theory of intermediate asymptotics \cite{BarenblattScaling}. Outside that regime no homogeneity is assumed.

\subsection{Assumptions}
\label{subsec:assumptions}
For later use it is convenient to collect the working hypotheses in a concise list. Throughout the scale-free analysis the following assumptions are imposed:
\begin{enumerate}[label=(\roman*),nosep]
  \item \emph{Quadratic structure and evenness in $Q$.} At fixed geometry the ledger is even and quadratic in $Q$,
        \[
          \mathcal P(Q,r) = A(r)\,Q^2 + B(r),
        \]
        with $A,B:(0,\infty)\to(0,\infty)$ of class $C^1$, as in Definition~\ref{def:admissible-ledger}.

  \item \emph{Locality along edges.} Edges are modeled as slender branches that contribute additively via per-length densities: a homogeneous edge of length $\ell$ with constant $(Q,r)$ carries cost $\mathcal C_{\text{branch}} = \ell\,\mathcal P(Q,r)$.

  \item \emph{Scale-free homogeneity on an open cone.} There exists a nonempty open set $D\subset(0,\infty)\times(0,\infty)$ (the scale-free regime) with the properties stated in Definition~\ref{def:admissible-ledger}(d), such that on $D$
        \[
          \mathcal P(\lambda^u Q,\lambda r)=\lambda^v \mathcal P(Q,r)
        \]
        whenever $(Q,r)$ and $(\lambda^u Q,\lambda r)$ both lie in $D$. This idealizes an intermediate asymptotic regime in the sense of Barenblatt \cite{BarenblattScaling}.

  \item \emph{Locally tree-like geometry without cancellation nodes.} The derivations of Murray-type closures, Young--Herring balances, and flux-only concavity are carried out on locally tree-like parts of the network without cycles. At degree-3 nodes attention is restricted to single-parent / two-daughter geometries with consistent flux orientation (no cancellation nodes).

  \item \emph{Piecewise straight edges for node translations.} In the node-translation variation each edge is treated as piecewise straight between nodes. Curvature-dependent corrections to length variations are higher order in the node displacement and do not affect the first-order Young--Herring-type balance.
\end{enumerate}

Items (i)–(iii) restate the ledger axioms in compact form; (iv)–(v) specify the geometric setting for the variational calculations. The next lemma is the technical core of the homogeneity analysis: under the strong assumption that $\mathcal P$ is exactly homogeneous on an open cone in $(Q,r)$, the ledger must, on that cone, reduce to a sum of two monomials in $r$.

\begin{lemma}[Euler homogeneity for admissible ledgers]
\label{lem:scale-free}
Let $\mathcal P(Q,r)$ be an admissible ledger in the sense of Definition~\ref{def:admissible-ledger}, with scale-free regime $D$ and associated connected projection interval $I$ on the $r$-axis as in item~(d). Then there exist constants $a,b>0$ and exponents $p,m\in\mathbb R$ such that, for all $(Q,r)\in D$ with $r\in I$,
\begin{equation}
  \mathcal P(Q,r)
  =
  a\,Q^2 r^{-p}
  +
  b\,r^{m}.
\end{equation}
Equivalently, within the scale-free regime and locally on this connected component $I$ of the $r$-axis, $\mathcal P$ can be written, up to relabeling of exponents, in the two-term power-law form used throughout this paper.
\end{lemma}

\begin{proof}
By assumption~(i), $\mathcal P(Q,r)=A(r)Q^2+B(r)$ with $A,B\in C^1((0,\infty))$ and $A,B>0$. Inserting this decomposition into the homogeneity condition
\eqref{eq:scale-free-homogeneity} yields, for all $(Q,r)\in D$ and all $\lambda>0$ such that $(\lambda^u Q,\lambda r)\in D$,
\begin{equation}
  A(\lambda r)\,\lambda^{2u} Q^2 + B(\lambda r)
  =
  \lambda^v \big[A(r)Q^2+B(r)\big].
\end{equation}
Fix $r$ in the projection interval $I$ from Definition~\ref{def:admissible-ledger}(d). By assumption, the slice
$J_r=\{Q>0:(Q,r)\in D\}$ contains a nontrivial open interval in $Q$, so the identity above holds for all $Q$ in such an interval. Since both sides are quadratic polynomials in $Q$, their coefficients must agree, which gives, for all $\lambda$ with $(\lambda^u Q,\lambda r)\in D$,
\begin{align}
  A(\lambda r)\,\lambda^{2u} &= \lambda^v A(r), \label{eq:A-scaling}\\
  B(\lambda r)              &= \lambda^v B(r). \label{eq:B-scaling}
\end{align}
Because $D$ is locally scale-invariant in the sense of Definition~\ref{def:admissible-ledger}(d)(i), for each fixed $(Q,r)\in D$ there exists $\varepsilon>0$ such that $(\lambda^u Q,\lambda r)\in D$ for all $\lambda\in(1-\varepsilon,1+\varepsilon)$. Thus, for each fixed $r\in I$, the relations \eqref{eq:A-scaling}–\eqref{eq:B-scaling} hold for all $\lambda$ in some open interval around $1$.

Differentiating \eqref{eq:B-scaling} with respect to $\lambda$ and evaluating at $\lambda=1$ gives, by the chain rule,
\[
  \frac{d}{d\lambda} B(\lambda r)\Big|_{\lambda=1}
  =
  r B'(r),
  \qquad
  \frac{d}{d\lambda}\bigl(\lambda^v B(r)\bigr)\Big|_{\lambda=1}
  =
  v B(r),
\]
so
\begin{equation}
  r B'(r) = v B(r)
  \qquad\text{for all } r\in I.
\end{equation}
Similarly, differentiating \eqref{eq:A-scaling} with respect to $\lambda$ and evaluating at $\lambda=1$ yields
\[
  \frac{d}{d\lambda}\bigl(A(\lambda r)\lambda^{2u}\bigr)\Big|_{\lambda=1}
  =
  r A'(r) + 2u A(r),
  \qquad
  \frac{d}{d\lambda}\bigl(\lambda^v A(r)\bigr)\Big|_{\lambda=1}
  =
  v A(r),
\]
hence
\begin{equation}
  r A'(r) + 2u A(r) = v A(r)
  \quad\Longrightarrow\quad
  r A'(r) = (v-2u)A(r)
  \qquad\text{for all } r\in I.
\end{equation}

On each connected component of $I$, these are first-order linear ODEs for $A$ and $B$. Since $A,B>0$ and are $C^1$, their solutions on $I$ are
\begin{equation}
  B(r) = b\,r^{v},
  \qquad
  A(r) = a\,r^{v-2u},
\end{equation}
for some constants $a,b>0$. Writing $m:=v$ and $p:=2u-v$ gives $v-2u=-p$ and hence
\[
  A(r) = a\,r^{-p},
  \qquad
  B(r) = b\,r^{m},
\]
so that, on the scale-free regime,
\begin{equation}
  \mathcal P(Q,r)
  =
  a\,Q^2 r^{-p} + b\,r^{m},
\end{equation}
which is the claimed form.
\end{proof}

\begin{remark}
Lemma~\ref{lem:scale-free} states that if the ledger is quadratic in $Q$ and exactly scale free on an open cone in $(Q,r)$ in the sense of Definition~\ref{def:admissible-ledger}(d), then on that cone the $Q$-dependent and $Q$-independent sectors must each be a \emph{pure} monomial in $r$. The proof only uses the homogeneity relations \eqref{eq:A-scaling}–\eqref{eq:B-scaling} for $\lambda$ near $1$: differentiability in $r$ and $\lambda$ forces $A$ and $B$ to satisfy Euler-type ODEs $rA'(r)=(v-2u)A(r)$ and $rB'(r)=vB(r)$, whose solutions are power laws. Logarithmic corrections (such as $Q^2 r^{-p}\log r$), additional intrinsic scales, or curvature- and anisotropy-dependent prefactors are all excluded by assumption, because any of these would spoil the exact local power-law homogeneity. In this sense Lemma~\ref{lem:scale-free} is simply Euler’s theorem applied separately to the one-variable homogeneous functions $A(r)$ and $B(r)$.

The lemma is not a claim that realistic ledgers have this form at all scales, but rather that on any idealised scale-free window an admissible quadratic ledger can be reduced to this homogeneous normal form. If the microscopic ledger includes additional powers of $r$ or slowly varying (for example logarithmic) prefactors, then the strong homogeneity axiom fails, and neither Lemma~\ref{lem:scale-free} nor the rigidity theorem in Sec.~\ref{sec:rigidity} applies literally; they should then be read as leading-order approximations within whatever scale range appears approximately homogeneous. The nontrivial constraints developed later arise from combining this monomial structure with branchwise optimization and node stationarity, so the two-term ledger class is best viewed as a convenient normal form once the homogeneity ansatz is adopted, not as an extra physical postulate.
\end{remark}

In physical applications it is natural to restrict attention further to $p>0$ and $m>0$, so that the transport term weakens with increasing radius and the structural term grows with radius. Within this restricted and exactly homogeneous quadratic class, Lemma~\ref{lem:scale-free} implies that any admissible ledger in a scale-free regime can, on that regime, be written in the compact form
\begin{equation}
  \label{eq:ledger-def}
  \mathcal P_{\rm br}(Q,r)
  =
  a\,\frac{Q^2}{r^{p}}
  +
  b\,r^{m},
  \qquad
  a,b>0,
  \quad
  p>0,
  \quad
  m>0.
\end{equation}
In what follows, all geometric interpretations of the Young--Herring balance and the symmetric Y-junction angle will implicitly be restricted to the subregime
\[
  0<m\le p,
\]
for which the angle law produces a real opening angle, and all statements about strictly concave Gilbert-type flux-only costs will be made in the stricter subregime
\[
  0<m<p
  \quad\Longleftrightarrow\quad
  0<\beta<1.
\]
Unless explicitly flagged as a boundary or out-of-regime example (such as $m=p$ in Sec.~\ref{sec:applications}), all examples below remain within these ranges. In the original entropy-per-information-cost formulation this two-term power law was interpreted as an EPIC ledger \cite{BennettEPIC}; in the present work~\eqref{eq:ledger-def} is simply referred to as the \emph{two-term (homogeneous) ledger class}. The remainder of the paper adopts~\eqref{eq:ledger-def} as a normal form for such idealized scale-free windows.

Given the ledger~\eqref{eq:ledger-def}, the total ledger functional on a graph configuration $\{Q_e,r_e\}$ is
\begin{equation}
  \label{eq:network-functional}
  \mathcal F[\{Q_e,r_e\}]
  =
  \sum_{e\in E} \ell_e\,\mathcal P_{\rm br}(Q_e,r_e)
  =
  \sum_{e\in E} \ell_e\left(a\,\frac{Q_e^2}{r_e^{p}}
  +
  b\,r_e^{m}\right).
\end{equation}
The geometric consequences of local optimality of $\mathcal F$ with respect to $\{r_e\}$ and to node positions, at fixed fluxes and boundary conditions, form the subject of Secs.~\ref{sec:murray}--\ref{sec:concavity}. Before turning to that analysis, it is useful to clarify how the abstract exponent $p$ relates to underlying transport equations and why a quadratic dependence on $Q$ is expected.

\subsection{Homogeneity as a selection principle and the transport exponent}
\label{subsec:EPIC-principle}

The axioms above can be viewed both as a convenient functional class and as the coarse-grained signature of a homogeneity-based selection principle. Many concrete models introduce the exponents that govern branched networks by solving detailed microscopic problems: Poiseuille flow with vessel upkeep for vascular trees, diffusion with latent-heat release for dendritic growth, or $\Gamma$-convergence of transport functionals for branched optimal transport. Additional extremum principles (such as MaxEP, MINEP, or optimizations over transport coefficients) have been proposed in related non-equilibrium contexts \cite{DewarMaxEP,MartyushevSeleznev,KleidonLorenz}, but the present work does \emph{not} propose an extra extremum principle that would single out particular values of $(m,p)$. Instead, the homogeneity relation~\eqref{eq:scale-free-homogeneity} and the quadratic dependence on $Q$ are taken as effective properties of a given scale-free regime, inferred from the underlying physics or from data; the rest of this section explores the geometric and energetic consequences of assuming such a homogeneous quadratic ledger, without claiming to derive the exponents $(m,p)$ from thermodynamics.

A useful perspective comes from linear irreversible thermodynamics. Consider a long, slender branch of length $\ell$ and effective radius $r$, with cross-section $A(r)$ and a scalar potential $\phi$ (pressure, temperature, voltage, chemical potential) obeying a linear, local transport law
\[
  \mathbf J = -\kappa\,\nabla\phi,
\]
with constant mobility or conductivity $\kappa>0$. For standard boundary conditions along a straight branch and no-flux on the lateral surface, the stationary solution has a uniform gradient along the branch, and the total flux through the cross-section is proportional to $A(r)$:
\[
  Q \propto -\,\kappa\,\frac{\Delta\phi}{\ell}\,A(r).
\]
The local entropy production density is $\sigma = T_{\rm env}^{-1}\,\mathbf J\cdot(-\nabla\phi)$, so the entropy production per unit length is proportional to
\[
  \frac{\dot S_{\rm prod}}{\ell}
  \propto
  \frac{\kappa A(r)}{T_{\rm env}}\left(\frac{\Delta\phi}{\ell}\right)^2
  \propto
  \frac{Q^2}{T_{\rm env}\,\kappa A(r)}.
\]
A standard derivation (see, for example, \cite{Onsager1931,DeGrootMazur,KondepudiPrigogine}) proceeds as follows. The key point is that, for any such linear-transport channel, the natural local ledger associated with dissipation,formalizes
\[
  \mathcal P_{\rm trans}(Q,r)
  \propto
  T_{\rm env}\frac{\dot S_{\rm prod}}{\ell},
\]
has the generic quadratic form
\begin{equation}
  \mathcal P_{\rm trans}(Q,r)
  \sim
  a\,\frac{Q^2}{A(r)},
  \label{eq:Q2-over-A}
\end{equation}
with $a>0$ absorbing $T_{\rm env}$ and material coefficients. If, in a given geometry, the cross-section scales as a power of radius, $A(r)\propto r^{p}$, the transport sector acquires the monomial form
\[
  \mathcal P_{\rm trans}(Q,r)
  \sim
  a\,\frac{Q^2}{r^{p}},
\]
with exponent $p$ set by the effective cross-sectional scaling of the underlying linear transport problem. Comparing \eqref{eq:Q2-over-A} with the transport sector $aQ^2/r^{p}$ in \eqref{eq:ledger-def}, $a$ may be identified (up to system-dependent constants) with $1/\kappa$, and $p$ can be interpreted as the exponent in the effective cross-section law $A(r)\propto r^{p}$.

Several standard examples fit this template. In Hagen--Poiseuille flow of a Newtonian fluid in a cylindrical tube, the Stokes equations yield an effective hydraulic conductance $\kappa_{\rm hyd}(r)\propto r^{4}$, so that the entropy production per unit length scales as $Q^{2}/r^{4}$ and $p=4$; in terms of \eqref{eq:Q2-over-A}, the role of $\kappa A(r)$ is played by the radius-dependent conductance $\kappa_{\rm hyd}(r)$. In diffusive or Ohmic transport through a cylindrical conductor in three dimensions, the cross-sectional area scales as $r^{2}$ and dissipation per unit length scales as $Q^{2}/r^{2}$, corresponding to $p=2$. Thin films, plates, or layered media can produce other effective exponents $p$ depending on how cross-sectional geometry scales with an effective radius.

Within this picture, the EPIC ledger \eqref{eq:ledger-def} can be regarded as a coarse-grained summary of two ingredients: a linear, local transport sector that fixes $p$ through the scaling of an effective cross-section $A(r)$, and a structural sector whose exponent $m$ encodes how upkeep cost scales with branch size (volume, surface area, or a more elaborate geometric measure). The analysis in the subsequent sections does not aim to derive $(m,p)$ from microscopic physics; instead, $(m,p)$ are treated as effective exponents that may be inferred from data or detailed simulations by measuring flux–radius scaling, concavity, and junction angles. Once a system has been assigned such effective exponents, the remaining results describe the geometric and energetic structure enforced by the EPIC ledger class: the generalized Murray law, the Young--Herring node balance and angle relation, the concave flux-only functional, and the single index $\chi=m/(m+p)$ that links them.

\section{Forward direction: consequences of the two-term ledger}
\label{sec:forward}

The two-term ledger \eqref{eq:ledger-def} now serves as the starting point for a ``forward'' analysis: the local consequences of its homogeneity for branch geometry and flux routing are derived in this section. The argument proceeds in four linked variational steps. First, minimizing the ledger with respect to radius at fixed flux produces a flux--radius power law and, at degree-3 nodes, a generalized Murray closure. Second, stationarity under translations of interior nodes yields a Young--Herring-type vector balance and fixes the opening angle of symmetric Y-junctions. Third, eliminating the radius recovers a concave flux-only functional of Gilbert type together with an equivalent flux-only node balance. Finally, a single dimensionless index $\chi$ organizes the flux--radius exponent, the concavity exponent, the Y-junction angle, and the internal split of the ledger into a compact algebraic dictionary.

\subsection{Flux--radius law and generalized Murray closure}
\label{sec:murray}

The discussion begins with the dependence of the ledger \eqref{eq:ledger-def} on branch radius at fixed flux. Optimizing the per-length cost for a single edge yields a flux--radius power law which, when combined with flux conservation at a degree-3 node, produces a generalized Murray closure. Both results are direct consequences of the homogeneous two-term structure of the ledger.

Consider a single branch of arclength $\ell>0$, radius $r>0$, and flux $Q\neq 0$. Its contribution to the ledger functional \eqref{eq:network-functional} is
\begin{equation}
  \ell\,\mathcal P_{\rm br}(Q,r)
  =
  \ell\left(a\,\frac{Q^2}{r^{p}} + b\,r^{m}\right).
\end{equation}
Since $\ell>0$ is a multiplicative constant, the optimal radius at fixed $Q$ is obtained by minimizing the per-length cost $\mathcal P_{\rm br}(Q,r)$ over $r>0$.

\begin{lemma}[Branchwise flux--radius scaling]
\label{lem:flux-radius}
Fix $Q\neq 0$. For the ledger \eqref{eq:ledger-def}, the map $r\mapsto \mathcal P_{\rm br}(Q,r)$ has a unique minimizer $r^\ast(Q)\in(0,\infty)$, and at that minimizer
\begin{equation}
  |Q| \propto \big(r^\ast\big)^{\alpha},
  \qquad
  \alpha := \frac{m+p}{2}.
\end{equation}
\end{lemma}

\begin{proof}
For fixed nonzero $Q$,
\[
  \mathcal P_{\rm br}(Q,r)
  =
  a\,Q^2 r^{-p} + b\,r^m,
  \qquad r>0.
\]
Differentiation with respect to $r$ gives
\[
  \frac{\partial \mathcal P_{\rm br}}{\partial r}
  =
  -p\,a\,Q^2 r^{-(p+1)}
  +
  m\,b\,r^{m-1}.
\]
Setting the derivative to zero yields
\[
  -p\,a\,Q^2 r^{-(p+1)} + m\,b\,r^{m-1} = 0
  \quad\Longleftrightarrow\quad
  p\,a\,Q^2 r^{-(p+1)} = m\,b\,r^{m-1}.
\]
Multiplying both sides by $r^{p+1}>0$ leads to
\[
  p\,a\,Q^2 = m\,b\,r^{m+p},
\]
and hence
\[
  Q^2 = \frac{m b}{p a}\,r^{m+p}.
\]
Taking square roots gives
\[
  |Q| = \sqrt{\frac{m b}{p a}}\;r^{(m+p)/2} \propto r^{\alpha},
  \qquad \alpha:=\frac{m+p}{2}.
\]

For uniqueness and minimality, it is convenient to rewrite
\[
  \frac{\partial \mathcal P_{\rm br}}{\partial r}
  =
  r^{-(p+1)}\left(-p\,a\,Q^2 + m\,b\,r^{m+p}\right).
\]
The prefactor $r^{-(p+1)}$ is positive for all $r>0$. The bracketed term
\[
  f(r):=-p\,a\,Q^2 + m\,b\,r^{m+p}
\]
is strictly increasing in $r$ because
\[
  f'(r) = (m+p)\,m\,b\,r^{m+p-1} > 0
  \qquad\text{for }r>0.
\]
Moreover, $\lim_{r\to 0^+} f(r) = -p\,a\,Q^2<0$ and $\lim_{r\to\infty} f(r)=+\infty$. Thus $f(r)$ crosses zero exactly once, from negative to positive, so $\partial_r \mathcal P_{\rm br}$ changes sign from negative to positive at a single radius. Consequently $\mathcal P_{\rm br}(Q,r)$ decreases, attains a unique minimum at $r=r^\ast(Q)$, and then increases, which proves the claim.
\end{proof}

\begin{remark}[Zero-flux edges]
Edges with $Q=0$ are degenerate for this optimization (the ledger is minimized only in a limit $r\to 0$ or $r\to\infty$ depending on $(m,p)$), so they are excluded from the branchwise analysis and play no role in the geometric relations derived below. In a global network they can still appear as dead-end branches or unused paths. Within the scale-free ledger picture, such zero-flux edges either shrink away in a fully optimal configuration (if the structural term dominates in the relevant regime) or else belong to a non-homogeneous part of the ledger where the simple two-term scale-free form~\eqref{eq:ledger-def} no longer applies. In particular, they are not part of the scale-free cone $D$ of Definition~\ref{def:admissible-ledger} on which the homogeneity arguments of Sec.~\ref{sec:branch-model} are carried out.
\end{remark}

Lemma~\ref{lem:flux-radius} shows that, at a branchwise optimum, the flux and radius along any edge with $Q_e\neq 0$ satisfy a simple power law. At the level of the full network, a natural optimization problem is
\[
  \min_{\{Q_e,r_e\}}\;\sum_{e\in E}\ell_e\,\mathcal P_{\rm br}(Q_e,r_e)
\]
subject to Kirchhoff flux constraints at interior nodes and prescribed boundary fluxes. Because the functional is a sum over edges and, for each fixed $Q_e$, the map $r_e\mapsto\mathcal P_{\rm br}(Q_e,r_e)$ is strictly convex with a unique minimizer, any \emph{global} minimizer $(\{Q_e,r_e\})$ must in particular minimize $\mathcal P_{\rm br}(Q_e,r_e)$ with respect to $r_e$ at fixed $Q_e$ on every edge. The two-stage procedure used here—first solve the branchwise problem at fixed $Q_e$, then minimize the resulting flux-only functional over $\{Q_e\}$—is therefore a convenient way of expressing \emph{necessary first-variation conditions} for $\mathcal F$ under radius variations and node translations. These conditions are not claimed to be sufficient to characterize global minimizers in the full configuration space, which would also require control over topological changes (edge creation/removal, reconnection, etc.). In physical terms, ``fixed $Q$'' in Lemma~\ref{lem:flux-radius} should be read as freezing, for the purposes of the local variation, the flux pattern induced by boundary conditions and nonlocal constraints; in a complete model the sets $\{Q_e\}$ and $\{r_e\}$ adjust self-consistently.

The branchwise scaling relation can be written as
\begin{equation}
  \label{eq:Q-r-alpha}
  |Q_e| = K\,r_e^\alpha,
  \qquad
  \alpha = \frac{m+p}{2},
\end{equation}
with a single positive constant
\[
  K = \sqrt{\frac{m b}{p a}}
\]
determined by the ledger parameters and common to all branches in a given regime. The sign of $Q_e$ is encoded in the orientation of the edge; the ledger depends only on $|Q_e|$.

\paragraph{Generalized Murray closure at junctions.}

The branchwise scaling relation \eqref{eq:Q-r-alpha} combines with flux conservation to produce a generalized Murray law at degree-3 junctions. To make contact with the classical setting, attention is restricted to a single \emph{parent} branch feeding two \emph{daughter} branches. Consider an interior node with three incident edges indexed by $i=0,1,2$. The analysis here is restricted to tree-like regions of the network (no local loops) and to the standard Murray geometry in which, in a suitable choice of edge orientations, all fluxes at the node share the same sign. Concretely, write
\[
  Q_i > 0, \qquad i=0,1,2,
\]
for the flux magnitudes and choose orientations so that flux conservation takes the form
\begin{equation}
  Q_0 = Q_1 + Q_2,
  \label{eq:node-flux-positive}
\end{equation}
with branch $0$ feeding branches $1$ and $2$. This excludes cancellation-type nodes in which one branch partially feeds another in the opposite direction, or nodes where flows reverse along an edge; such configurations are not the target of the Murray-type closure derived here. In the remainder of this subsection, a ``degree-3 Murray node'' will always mean a node with this single-parent, two-daughter, no-cancellation structure.

Applying the flux--radius law \eqref{eq:Q-r-alpha} to each edge gives
\[
  Q_i = K r_i^\alpha,
  \qquad i=0,1,2.
\]
Substituting into \eqref{eq:node-flux-positive} and canceling $K$ produces
\[
  r_0^\alpha = r_1^\alpha + r_2^\alpha.
\]

\begin{corollary}[Generalized Murray law]
\label{cor:murray}
Let $v$ be a degree-3 node with one parent edge $0$ and two daughter edges $1,2$, and assume that there is a consistent choice of edge orientations such that all three fluxes at $v$ are strictly positive and obey
\begin{equation}
  Q_0 = Q_1 + Q_2,
  \qquad
  Q_i>0,\; i=0,1,2.
\end{equation}
Equivalently, $v$ is a single-parent, two-daughter, no-cancellation node in the sense described above. Then, at a branchwise optimum of the ledger \eqref{eq:ledger-def}, the corresponding radii satisfy
\begin{equation}
  r_0^\alpha
  =
  r_1^\alpha + r_2^\alpha,
  \qquad
  \alpha=\frac{m+p}{2}.
\end{equation}
\end{corollary}

Several classical cases are recovered as particular points in the $(m,p)$ plane. For $(p,m)=(4,2)$, corresponding to Poiseuille transport with volume-priced upkeep, the exponent is $\alpha=3$ and the relation reduces to Murray's classical law $r_0^{3}=r_1^{3}+r_2^{3}$. For $(p,m)=(4,1)$, appropriate to surface-priced upkeep, $\alpha=5/2$; for $(p,m)=(2,2)$, corresponding to a conduction-like transport term with $p=2$ and volume-based upkeep, the closure becomes area conservation $r_0^2 = r_1^2 + r_2^2$. The essential point is that, whenever the single-parent two-daughter, same-sign flux assumption holds and the ledger is in the two-term class~\eqref{eq:ledger-def}, any choice of exponents $(m,p)$ produces a corresponding generalized Murray exponent $\alpha=(m+p)/2$ and a node closure of the form of Corollary~\ref{cor:murray}. In this sense, the usual cubic Murray law and its sub- and super-cubic generalizations are geometric manifestations of the underlying two-term homogeneous ledger.

\subsection{Node translations and Young--Herring angle law}
\label{sec:angles}
The next variational step concerns the dependence of the ledger functional on the spatial position of interior nodes. Fluxes and radii are held at their branchwise optima, and the response of the total cost to infinitesimal translations of a node is examined. Stationarity under such translations yields a vector balance at each junction, identical in structure to a Young--Herring triple-junction condition. For symmetric Y-junctions this balance leads to a closed-form expression for the opening angle in terms of the ledger exponents $(m,p)$. Throughout this subsection a restricted class of variations is considered in which node positions are perturbed, $\delta\bm x_v\neq 0$, while the fluxes and radii on the incident edges are held fixed, $\delta Q_e=\delta r_e=0$. This is the standard Steiner-type setup and is justified because the along-edge radius optimization from Sec.~\ref{sec:murray} decouples from edge length, and because the flux pattern is treated as imposed by distant boundary conditions on the timescale of node motion. The Young--Herring law derived below should therefore be read as a local necessary condition for stationarity under these node translations; a fully coupled optimization in which $\{Q_e,r_e\}$ are also allowed to adjust would in general involve additional terms beyond the length-variation contribution analyzed here.

To make this precise, let $v\in V$ be an interior node and let $E(v)$ denote the set of edges incident at $v$. For each $e\in E(v)$ let $\ell_e$ denote the distance from $v$ to a fixed far-field point along that edge, so that $\ell_e$ changes if the node is displaced; let $r_e$ and $Q_e$ denote the corresponding radius and flux. Denote by $\hat{\bm t}_e$ the unit tangent along edge $e$ at $v$, oriented away from the node. Each edge is regarded as a smooth embedded curve parametrized by arclength, so that the first variation of its length with respect to a perturbation of the common endpoint is given by the unit tangent at that endpoint. Equivalently, the network may be treated as piecewise straight between nodes; curvature effects then enter only at second order in the node displacement and do not affect the first-variation balance derived below. In the node-translation calculation the fluxes $Q_e$ and radii $r_e$ are held fixed, and any global constraints (such as boundary fluxes or a fixed total material budget) are assumed to be unaffected by the local translation.

The contribution of the edges incident at $v$ to the ledger functional is
\begin{equation}
  \mathcal F_v
  =
  \sum_{e\in E(v)} \ell_e\,\mathcal P_{\rm br}(Q_e,r_e).
\end{equation}
Under an infinitesimal translation of the node,
\[
  X(v)\mapsto X(v)+\delta\bm x_v,
\]
the lengths change, to first order, as
\begin{equation}
  \delta \ell_e = -\,\hat{\bm t}_e\cdot\delta\bm x_v,
  \qquad e\in E(v),
\end{equation}
while $Q_e$ and $r_e$ are held fixed (a displacement of the node along the edges without changing their cross-sections or far-field endpoints). The first variation of $\mathcal F_v$ is therefore
\begin{equation}
  \delta \mathcal F_v
  =
  \sum_{e\in E(v)} \mathcal P_{\rm br}(Q_e,r_e)\,\delta\ell_e
  =
  -\,\delta\bm x_v\cdot
  \sum_{e\in E(v)} \mathcal P_{\rm br}(Q_e,r_e)\,\hat{\bm t}_e.
\end{equation}

Stationarity with respect to arbitrary $\delta\bm x_v$ requires
\begin{equation}
  \label{eq:raw-node-balance}
  \sum_{e\in E(v)} \mathcal P_{\rm br}(Q_e,r_e)\,\hat{\bm t}_e = \bm 0.
\end{equation}

In configurations that are already optimal with respect to $r_e$ at fixed $Q_e$, the branchwise condition of Lemma~\ref{lem:flux-radius} can be used to eliminate $Q_e$ in favour of $r_e$. From $Q_e^2=(m b/(p a))\,r_e^{m+p}$ one finds
\[
  a\,\frac{Q_e^2}{r_e^p}
  =
  a\,\frac{m b}{p a}\,r_e^{m+p-p}
  =
  \frac{m b}{p}\,r_e^m.
\]
Consequently, on each edge
\begin{equation}
  \mathcal P_{\rm br}(Q_e,r_e)
  =
  a\,\frac{Q_e^2}{r_e^{p}}
  +
  b\,r_e^m
  =
  \left(\frac{m b}{p} + b\right) r_e^m
  =
  \kappa\,r_e^m,
\end{equation}
with a node-independent prefactor
\[
  \kappa = b\Big(1+\frac{m}{p}\Big)>0.
\]
Substituting into \eqref{eq:raw-node-balance} yields the compact vector closure
\begin{equation}
  \label{eq:bit-tension}
  \sum_{e\in E(v)} r_e^m\,\hat{\bm t}_e = \bm 0.
\end{equation}

Relation \eqref{eq:bit-tension} admits a natural mechanical interpretation. For fixed $Q_e$ and $r_e$, the per-length contribution of edge $e$ to the ledger is $\mathcal P_{\rm br}(Q_e,r_e)$, so the derivative of $\mathcal F_v$ with respect to a change in the length of that edge at fixed $(Q_e,r_e)$ is simply
\[
  \frac{\partial}{\partial \ell_e}\big(\ell_e \mathcal P_{\rm br}(Q_e,r_e)\big)
  = \mathcal P_{\rm br}(Q_e,r_e)
  = \kappa r_e^m.
\]
Thus, in the branchwise optimal regime, $r_e^m$ weights the tangents in \eqref{eq:bit-tension} in the same way that line tensions weight interface normals in a classical Young--Herring law. Introducing the effective skeleton tensions
\begin{equation}
  \tau_e^{\rm eff} := \kappa r_e^m,
\end{equation}
equation \eqref{eq:bit-tension} can be rewritten as
\begin{equation}
  \sum_{e\in E(v)} \tau_e^{\rm eff}\,\hat{\bm t}_e = \bm 0.
  \label{eq:EPIC-Young-Herring}
\end{equation}
Equation~\eqref{eq:EPIC-Young-Herring} is therefore a Young--Herring triple-junction condition at the level of the coarse-grained branch skeleton \cite{Herring1951,CahnHilliard1958}: the directions of the branches at a node are weighted by effective tensions derived from the EPIC ledger rather than from microscopic surface free energies.

\paragraph{Symmetric Y-junction and opening angle.}

The vector balance \eqref{eq:bit-tension} determines the opening angle at symmetric Y-junctions. The analysis remains in the single-parent, two-daughter, no-cancellation setting of Sec.~\ref{sec:murray}, so that, in a consistent orientation, the parent carries the sum of the daughter fluxes with all three fluxes having the same sign. Consider such a node with one parent edge (index $0$) and two daughter edges (indices $1,2$) of equal radii $r_1=r_2=:r_d$, meeting at opening angle $\theta$ symmetric about the parent axis. Choose coordinates so that the parent branch points along the negative $x$-axis:
\[
  \hat{\bm t}_0 = (-1,0),
\]
and the daughters are symmetric about the $x$-axis:
\[
  \hat{\bm t}_1 = \big(\cos\tfrac{\theta}{2},\sin\tfrac{\theta}{2}\big),
  \qquad
  \hat{\bm t}_2 = \big(\cos\tfrac{\theta}{2},-\sin\tfrac{\theta}{2}\big).
\]
The $y$-component of \eqref{eq:bit-tension} vanishes by symmetry. The $x$-component reads
\[
  r_0^m(-1) + r_d^m\cos\tfrac{\theta}{2} + r_d^m\cos\tfrac{\theta}{2} = 0,
\]
that is,
\begin{equation}
  -r_0^m + 2 r_d^m \cos\frac{\theta}{2} = 0.
\end{equation}
Solving for $\cos(\theta/2)$ gives
\begin{equation}
  \cos\frac{\theta}{2}
  =
  \frac{r_0^m}{2 r_d^m}.
  \label{eq:cos-theta-radius}
\end{equation}

For symmetric Y-junctions the generalized Murray law of Corollary~\ref{cor:murray} implies
\[
  r_0^\alpha = 2 r_d^\alpha,
\]
with $\alpha=(m+p)/2$. Raising both sides to the power $m/\alpha$ yields
\[
  r_0^m = \big(2 r_d^\alpha\big)^{m/\alpha} = 2^{m/\alpha} r_d^m.
\]
Substituting into \eqref{eq:cos-theta-radius} gives
\begin{equation}
  \cos\frac{\theta}{2}
  =
  2^{\frac{m}{\alpha}-1}.
  \label{eq:theta-m-alpha}
\end{equation}
Using $\alpha=(m+p)/2$ from Lemma~\ref{lem:flux-radius}, this can be rewritten as
\begin{equation}
  \cos\frac{\theta}{2}
  =
  2^{\frac{2m}{m+p}-1}
  =
  2^{\frac{m-p}{m+p}}.
  \label{eq:theta-mp}
\end{equation}
For $0<m\le p$ the exponent $(m-p)/(m+p)$ is non-positive, so that $0<\cos(\theta/2)\le 1$ and a real symmetric Y-junction angle exists. For $m>p$ the exponent is positive and the right-hand side of \eqref{eq:theta-mp} exceeds $1$, signaling that the simple symmetric Y-geometry is no longer compatible with the homogeneous ledger in that parameter range. Accordingly, geometric interpretations of $\theta$ will be restricted to $0<m\le p$. The limiting cases are instructive:
\begin{itemize}
  \item As $m\to 0^+$ (radius-independent structural pricing),
        \[
          \frac{m-p}{m+p}\to -1,
          \qquad
          \cos\frac{\theta}{2}\to \frac{1}{2},
        \]
        so $\theta\to 120^\circ$, the Steiner angle for minimal-length trees with equal line tensions \cite{MorganGMT}.
  \item As $m\to p^-$ (transport and structural exponents coincide),
        \[
          \frac{m-p}{m+p}\to 0,
          \qquad
          \cos\frac{\theta}{2}\to 1,
        \]
        so $\theta\to 0^\circ$ and the junction degenerates.
\end{itemize}

Thus, once the ledger exponents $(m,p)$ are specified, the symmetric Y-junction angle is fixed; no additional geometric parameter is available at the node. The law \eqref{eq:theta-mp} therefore plays the role of a Snell-type refraction relation at junctions, with the effective ``refractive properties'' encoded in $(m,p)$. In Sec.~\ref{sec:chi-dictionary} this relation is re-expressed in terms of the concavity exponent $\beta$ of the flux-only ledger as
\[
  \cos\frac{\theta}{2} = 2^{\beta-1},
\]
which is precisely the classical Snell-type angle law for symmetric triple junctions in Gilbert/branched-transport models, obtained by minimizing the functional $\sum_e \ell_e |Q_e|^\beta$ and enforcing the flux-only balance $\sum_e |Q_e|^\beta \hat{\bm t}_e=0$; see, for example, \cite{Gilbert1967,BernotCasellesMorel,XiaBranchedTransport}. Within the homogeneous two-term ledger class, that flux-only angle law is therefore not an independent phenomenological assumption: the same exponents $(m,p)$ that control the flux--radius scaling and the concavity of the effective cost also determine the symmetric opening angle via $\chi$.

\subsection{Effective flux-only cost and concavity}
\label{sec:concavity}

The previous subsections showed that branchwise optimality of the two-term ledger~\eqref{eq:ledger-def} produces a power-law flux--radius relation and that node translations enforce a Young--Herring-type vector balance with radius weights $r^m$. The remaining local step is to eliminate the radius altogether and obtain an effective flux-only cost. This yields a single exponent $\beta$ that governs the concavity of the network functional and connects the present framework directly to classical branched-transport functionals of Gilbert type.

\begin{lemma}[Effective edge cost and concavity]
\label{lem:beta}
Let $r^\ast(Q)$ denote the branchwise minimizer of the ledger $\mathcal P_{\rm br}(Q,r)$ from Lemma~\ref{lem:flux-radius}. Then the minimized per-length edge cost
\begin{equation}
  \mathcal P_{\rm br}^\ast(Q)
  :=
  \min_{r>0} \mathcal P_{\rm br}(Q,r)
\end{equation}
scales as a power of $|Q|$,
\begin{equation}
  \mathcal P_{\rm br}^\ast(Q) \propto |Q|^\beta,
  \qquad
  \beta := \frac{2m}{m+p}.
\end{equation}
For $0<m<p$ one has $0<\beta<1$, and the flux-only cost is strictly concave in $|Q|$.
\end{lemma}

\begin{proof}
From Lemma~\ref{lem:flux-radius}, the branchwise minimizer satisfies a power-law scaling
\[
  r^\ast(Q)\propto |Q|^{2/(m+p)}.
\]
More explicitly, there exists a constant $C>0$ such that
\[
  r^\ast(Q) = C\,|Q|^{2/(m+p)}.
\]
Substituting this into the ledger~\eqref{eq:ledger-def} gives
\begin{align}
  \mathcal P_{\rm br}^\ast(Q)
  &=
  a\,\frac{Q^2}{(r^\ast)^p} + b\,(r^\ast)^m \nonumber\\
  &=
  a\,Q^2 C^{-p}|Q|^{-2p/(m+p)}
  +
  b\,C^m |Q|^{2m/(m+p)} \nonumber\\
  &=
  a C^{-p}|Q|^{2-2p/(m+p)}
  +
  b C^{m}|Q|^{2m/(m+p)}.
\end{align}
The two exponents coincide, since
\[
  2 - \frac{2p}{m+p}
  =
  \frac{2(m+p) - 2p}{m+p}
  =
  \frac{2m}{m+p}.
\]
Both terms therefore scale as $|Q|^{2m/(m+p)}$. Writing
\[
  \beta := \frac{2m}{m+p},
\]
one obtains
\[
  \mathcal P_{\rm br}^\ast(Q) \propto |Q|^\beta.
\]

The map $q\mapsto |q|^\beta$ is strictly concave on $\mathbb R\setminus\{0\}$ whenever $0<\beta<1$, since its second derivative is negative away from the origin. For $0<m<p$ one indeed has $0<\beta<1$, so $\mathcal P_{\rm br}^\ast(Q)$ is strictly concave in $|Q|$ on this parameter range.
\end{proof}

Once the radius has been eliminated in this way, the network functional takes the canonical branched-transport form
\begin{equation}
  \mathcal F^\ast[\{Q_e\}]
  \propto
  \sum_{e\in E} \ell_e |Q_e|^\beta,
  \qquad
  \beta = \frac{2m}{m+p},
  \label{eq:flux-only-Gilbert}
\end{equation}
with the exponent $\beta$ fixed by the ledger exponents $(m,p)$ through Lemma~\ref{lem:beta}. For $0<m<p$ one has $0<\beta<1$, so the flux-only cost is strictly concave and the cost of carrying flux is strictly subadditive: for $Q_1,Q_2>0$,
\[
  \mathcal P_{\rm br}^\ast(Q_1+Q_2)
  <
  \mathcal P_{\rm br}^\ast(Q_1)+\mathcal P_{\rm br}^\ast(Q_2).
\]
The boundary case $m=p$ gives $\beta=1$ and a \emph{linear} flux-only cost $|Q|$, which marks the critical borderline between strictly concave Gilbert-type behavior and linear transport. For $0<m<p$, merging fluxes into a common edge is therefore energetically beneficial, which is the mechanism behind trunk sharing and deep tree hierarchies in branched-transport networks.

Promoting the $Q_e$ to primary variables and viewing~\eqref{eq:flux-only-Gilbert} as a cost functional yields precisely the classical Gilbert-type branched-transport functional $\sum_e \ell_e |Q_e|^\beta$ studied in optimal network theory \cite{Gilbert1967,BernotCasellesMorel,XiaBranchedTransport,BouchitteButtazzo,AmbrosioFuscoPallara}. The strict concavity condition $\beta<1$ encapsulates the energetic preference for routing flux through shared trunks instead of through a collection of disjoint paths.

At the node level, the radius-based Young--Herring balance~\eqref{eq:bit-tension} and the flux--radius scaling $|Q_e|\propto r_e^\alpha$ with $\alpha=(m+p)/2$ together imply an equivalent flux-only balance. Eliminating $r_e$ via $r_e\propto |Q_e|^{1/\alpha}$ and using $\beta=2m/(m+p)$ shows that the vector condition
\[
  \sum_{e\in E(v)} r_e^m\,\hat{\bm t}_e = \bm 0
\]
can be rewritten, up to an overall constant factor, as
\[
  \sum_{e\in E(v)} |Q_e|^\beta \hat{\bm t}_e = \bm 0.
\]
This recovers the flux-weighted Young--Herring node law associated with the flux-only functional~\eqref{eq:flux-only-Gilbert} in the branched-transport literature; see, for example, \cite{Gilbert1967,XiaBranchedTransport,BernotCasellesMorel}. When applied to a symmetric degree-3 node it reproduces the standard Snell-type relation $\cos(\theta/2)=2^{\beta-1}$. Within the homogeneous two-term ledger class, the flux-only node balance, the exponent $\beta$, and the associated angle law are therefore \emph{not} additional phenomenological assumptions: they all follow from the same underlying ledger~\eqref{eq:ledger-def}, with $\beta$ algebraically tied to $(m,p)$ and hence to the generalized Murray law and the symmetric junction angle via the EPIC index $\chi$.

\subsection{Single-index dictionary and ledger-sector split}
\label{sec:chi-dictionary}

The previous subsections identified three quantities that govern the geometry of locally optimal networks in the two-term ledger class~\eqref{eq:ledger-def}: the flux--radius exponent $\alpha$, the concavity exponent $\beta$ of the effective flux-only cost, and the symmetric Y-junction opening angle $\theta$. Each arises from the same pair of ledger exponents $(m,p)$, but that parametrization obscures how these observables constrain one another. This subsection introduces a single dimensionless index $\chi$ that organizes $\alpha$, $\beta$, and $\theta$ into a compact algebraic dictionary and then shows that the same index quantifies the split between transport and structural contributions inside the ledger itself.

\paragraph{EPIC index and Snell--Murray identities.}

From Lemma~\ref{lem:flux-radius}, Lemma~\ref{lem:beta}, and the angle relation~\eqref{eq:theta-mp}, the exponents extracted from the ledger~\eqref{eq:ledger-def} obey
\begin{equation}
  \alpha = \frac{m+p}{2},
  \qquad
  \beta  = \frac{2m}{m+p},
  \qquad
  \cos\frac{\theta}{2} = 2^{\frac{m-p}{m+p}}.
\end{equation}
All three are therefore functions of the two ledger exponents $(m,p)$. A more transparent parametrization is obtained by introducing the ratio
\begin{equation}
  \chi
  :=
  \frac{m}{m+p}
  \in (0,1),
\end{equation}
referred to here simply as the \emph{index} $\chi$. By construction,
\begin{equation}
  \beta
  =
  \frac{2m}{m+p}
  =
  2\chi.
\end{equation}
Solving $\alpha=(m+p)/2$ for $m$ and $p$ in terms of $(\alpha,\chi)$ gives
\[
  m+p = 2\alpha,
  \qquad
  m = \chi(m+p) = 2\alpha\chi,
\]
and hence
\begin{equation}
  m = 2\alpha\chi,
  \qquad
  p = 2\alpha(1-\chi),
\end{equation}
so that
\begin{equation}
  \chi
  =
  \frac{m}{m+p}
  =
  1 - \frac{p}{2\alpha}.
\end{equation}
For the opening angle, the exponent in~\eqref{eq:theta-mp} can be written
\begin{equation}
  \frac{m-p}{m+p}
  =
  1 - 2\frac{p}{m+p}
  =
  1 - 2(1-\chi)
  =
  2\chi-1,
\end{equation}
which yields
\begin{equation}
  \cos\frac{\theta}{2}
  =
  2^{2\chi-1},
  \qquad
  \chi
  =
  \frac{1}{2}\Big(1 + \log_2\cos\tfrac{\theta}{2}\Big).
\end{equation}
These identities can be collected in a single dictionary.

\begin{theorem}[Single-index dictionary within the two-term ledger class]
\label{thm:chi-dictionary}
Consider a branch network governed, in a scale-free regime, by the ledger~\eqref{eq:ledger-def} with $a,b>0$ and $m,p>0$, and a configuration $\{Q_e,r_e\}$ that is locally optimal with respect to edge radii and stationary with respect to interior node translations, under flux conservation and fixed boundary conditions. Define:
\begin{itemize}
  \item $\alpha$ by the branchwise flux--radius law $|Q|\propto r^\alpha$,
  \item $\beta$ by the effective edge cost $\mathcal P_{\rm br}^\ast(Q)\propto |Q|^\beta$,
  \item $\theta$ as the symmetric Y-junction opening angle.
\end{itemize}
Then the EPIC index
\begin{equation}
  \chi := \frac{m}{m+p}\in(0,1)
\end{equation}
admits the equivalent representations
\begin{equation}
  \boxed{
    \chi
    =
    1 - \frac{p}{2\alpha}
    =
    \frac{\beta}{2}
    =
    \frac{1}{2}\Big(1 + \log_2\cos\tfrac{\theta}{2}\Big),
  }
\end{equation}
and, provided $0<m\le p$ so that a real symmetric Y-angle exists,
\begin{equation}
  \boxed{
    \cos\frac{\theta}{2}
    =
    2^{\,\beta-1}
    =
    2^{\,1-\frac{p}{\alpha}},
  }
\end{equation}
with algebraic closures
\begin{equation}
  \alpha = \frac{p}{2-\beta},
  \qquad
  \beta = 2 - \frac{p}{\alpha}.
\end{equation}
Equivalently, for arbitrary $m,p>0$ the triple $(\alpha,\beta,\chi)$ is simply a reparameterization of $(m,p)$. The extension to the quadruple $(\alpha,\beta,\theta,\chi)$, and hence to the Snell--Murray dictionary involving the symmetric Y-angle and the Gilbert-type flux-only functional, is valid only in the parameter range $0<m\le p$ (with $0<m<p$ for a strictly concave cost $0<\beta<1$). In that range, specifying any two of $\{\alpha,\beta,\theta\}$ together with $p$ determines the EPIC index $\chi$ and all remaining quantities.
\end{theorem}

In particular, the angular law
\[
  \cos\frac{\theta}{2} = 2^{\beta-1}
\]
is exactly the Snell-type condition for symmetric Y-junctions in Gilbert networks with cost functional $\sum_e \ell_e |Q_e|^\beta$ and node balance $\sum_e |Q_e|^\beta \hat{\bm t}_e = \bm 0$ \cite{Gilbert1967,XiaBranchedTransport,BernotCasellesMorel}. The novelty in the present setting is that, once the ledger is restricted to the two-term scale-free form~\eqref{eq:ledger-def}, the parameter $\beta$ ceases to be an independent phenomenological input: it is fixed by $(m,p)$, or equivalently by $\chi$, and is algebraically tied to the flux--radius law and the junction angle. Table~\ref{tab:dictionary} summarizes these relations in the concave Gilbert regime $0<m<p$.
\begin{table}[t]
  \caption{Algebraic relations between ledger exponents $(m,p)$ and network observables $(\alpha,\beta,\theta,\chi)$ in the concave Gilbert regime $0<m<p$. Here $\alpha$ is the flux--radius exponent, $\beta$ the concavity exponent of the effective flux-only cost, $\theta$ the symmetric Y-junction opening angle, and $\chi$ the single index introduced in
  Sec.~\ref{sec:chi-dictionary}.}
  \label{tab:dictionary}
  \centering

  \scriptsize
  \setlength{\tabcolsep}{2pt}

  \begin{tabular}{llll}
    \toprule
    Quantity & Definition & Range & Interpretation \\
    \midrule
    $(m,p)$
      & given structural and transport exponents
      & $0<m<p$
      & powers of $r$ in the ledger \\
    $\alpha$
      & $\alpha = \dfrac{m+p}{2}$
      & $>0$
      & flux--radius law $|Q|\propto r^\alpha$; Murray closure \\
    $\beta$
      & $\beta = \dfrac{2m}{m+p} = 2\chi$
      & $0<\beta<1$
      & concavity exponent of $\sum_e \ell_e |Q_e|^\beta$ \\
    $\chi$
      & $\chi = \dfrac{m}{m+p} = 1-\dfrac{p}{2\alpha} = \dfrac{\beta}{2}$
      & $0<\chi<\tfrac12$
      & fraction of the ledger spent on transport at the branchwise optimum \\
    $\theta$
      & $\cos(\theta/2) = 2^{\beta-1} = 2^{(m-p)/(m+p)}$
      & $0<\theta<120^\circ$
      & symmetric Y-junction angle; $\theta\to 120^\circ$ as $\chi\to 0$ \\
    \bottomrule
  \end{tabular}
\end{table}

\paragraph{Ledger-sector decomposition and meaning of $\chi$.}

The index $\chi$ arises naturally as a reparametrization of $(m,p)$, but it also has a direct interpretation inside the ledger itself. At the branchwise optimum, $\chi$ measures the fraction of the local ledger devoted to transport versus structure.

Consider the two-term ledger~\eqref{eq:ledger-def} on a single branch and let $r^\ast(Q)$ denote the optimal radius from Lemma~\ref{lem:flux-radius}. At that optimum, the transport and structural contributions are
\begin{align}
  \mathcal P_{\rm trans}
  &:= a\,\frac{Q^2}{(r^\ast)^p}, \\
  \mathcal P_{\rm struct}
  &:= b\,(r^\ast)^m,
\end{align}
so that the total per-length ledger is
\[
  \mathcal P_{\rm br}^\ast(Q)
  =
  \mathcal P_{\rm trans} + \mathcal P_{\rm struct}.
\]
Using the optimality condition from Lemma~\ref{lem:flux-radius},
\[
  Q^2 = \frac{m b}{p a}\,(r^\ast)^{m+p},
\]
the transport term can be rewritten as
\begin{equation}
  \mathcal P_{\rm trans}
  =
  a\,\frac{Q^2}{(r^\ast)^p}
  =
  a\,\frac{m b}{p a}\,(r^\ast)^m
  =
  \frac{m b}{p}\,(r^\ast)^m,
\end{equation}
while the structural term remains
\begin{equation}
  \mathcal P_{\rm struct}
  =
  b\,(r^\ast)^m.
\end{equation}
The total optimal ledger is therefore
\begin{equation}
  \mathcal P_{\rm br}^\ast(Q)
  =
  \left(\frac{m b}{p} + b\right)(r^\ast)^m
  =
  b\,\frac{m+p}{p}\,(r^\ast)^m.
\end{equation}
The fractions of the ledger allocated to transport and structure at the optimum are
\begin{align}
  \frac{\mathcal P_{\rm trans}}{\mathcal P_{\rm br}^\ast}
  &=
  \frac{(m b/p)(r^\ast)^m}{b[(m+p)/p](r^\ast)^m}
   =
   \frac{m}{m+p}, \\
  \frac{\mathcal P_{\rm struct}}{\mathcal P_{\rm br}^\ast}
  &=
  \frac{b(r^\ast)^m}{b[(m+p)/p](r^\ast)^m}
   =
   \frac{p}{m+p}.
\end{align}
In terms of the EPIC index this becomes
\begin{equation}
  \frac{\mathcal P_{\rm trans}}{\mathcal P_{\rm br}^\ast}
  =
  \chi,
  \qquad
  \frac{\mathcal P_{\rm struct}}{\mathcal P_{\rm br}^\ast}
  =
  1-\chi.
  \label{eq:chi-fraction}
\end{equation}

\begin{lemma}[Ledger-sector decomposition]
\label{lem:chi-fraction}
For the two-term ledger~\eqref{eq:ledger-def}, at the branchwise optimal radius $r^\ast(Q)$ the EPIC index $\chi=m/(m+p)$ equals the fraction of the branch ledger spent on the transport term:
\begin{equation}
  \chi
  =
  \frac{\mathcal P_{\rm trans}}{\mathcal P_{\rm br}^\ast},
  \qquad
  1-\chi
  =
  \frac{\mathcal P_{\rm struct}}{\mathcal P_{\rm br}^\ast}.
\end{equation}
\end{lemma}

From a mathematical standpoint, Lemma~\ref{lem:chi-fraction} is a direct consequence of homogeneity: the ledger is the sum of two homogeneous monomials in $r$, and at the minimum their relative weights are fixed by their degrees. The index $\chi$ depends only on the exponents $m$ and $p$: multiplying the entire ledger by a positive constant, or rescaling $Q$ and $r$ by fixed factors (which only changes the prefactors $a,b$), leaves both $\chi$ and the ratio $\mathcal P_{\rm trans}/\mathcal P_{\rm br}^\ast$ unchanged. In that sense $\chi$ is insensitive to the overall choice of units or to the normalization of the ledger. For the interpretation of $\chi$ as a \emph{fraction of a budget}, however, one must first express the transport and structural terms in commensurate units (in the EPIC setting, both are written in J/bit per unit length, so that $T_{\rm env}\dot S_{\rm prod}$ and $\varepsilon_b\dot I_{\rm struct}$ can be added as a single scalar); otherwise $\chi$ still encodes the relative exponents, but the literal numerical split between ``transport'' and ``structure'' depends on the calibration.

In physical terms, the limiting cases are instructive:
\begin{itemize}[nosep]
  \item $\chi\to 0$: the structural sector dominates, $\mathcal P_{\rm struct}\gg \mathcal P_{\rm trans}$ at the optimum.
  \item $\chi=\tfrac{1}{2}$: transport and structure share the ledger equally, $\mathcal P_{\rm trans}=\mathcal P_{\rm struct}$.
  \item $\chi\to 1$: the transport sector dominates, $\mathcal P_{\rm trans}\gg \mathcal P_{\rm struct}$.
\end{itemize}
Thus, at the branchwise optimum, $\chi$ is a dimensionless fraction between $0$ and $1$ that quantifies how the local ledger is split between transport and structure, independently of the overall choice of units or normalization of $\mathcal P$.

Combining Theorem~\ref{thm:chi-dictionary} with Lemma~\ref{lem:chi-fraction} yields a compact summary:
\begin{equation}
  \chi
  =
  \underbrace{\frac{\mathcal P_{\rm trans}}{\mathcal P_{\rm br}^\ast}}_{\text{local ledger split}}
  =
  \underbrace{\frac{\beta}{2}}_{\text{network concavity}}
  =
  \underbrace{1-\frac{p}{2\alpha}}_{\text{flux--radius scaling}}
  =
  \underbrace{\frac{1}{2}\Big(1+\log_2\cos\tfrac{\theta}{2}\Big)}_{\text{junction angle}}.
\end{equation}
Within the homogeneous two-term ledger class, this single scalar $\chi$ therefore encodes simultaneously:
\begin{enumerate}[label=(\roman*),nosep]
  \item the split between transport and structural contributions in the local ledger at the branchwise optimum,
  \item the concavity exponent of the effective flux-only cost on the network,
  \item the flux--radius scaling and generalized Murray closure at junctions,
  \item and the symmetric opening angle at Y-junctions.
\end{enumerate}
Local budgeting, network concavity, flux scaling, and junction geometry are thus not independent features of a branched network in this class; they are different manifestations of the same underlying EPIC index~$\chi$.

\begin{remark}
Theorem~\ref{thm:chi-dictionary} provides an internal algebraic dictionary for the two-term, scale-free ledger class, not a universality statement in the renormalization-group sense. Once the ledger has the form~\eqref{eq:ledger-def}, the exponents $(\alpha,\beta,\theta)$, the ledger exponents $(m,p)$, and the index $\chi$ are merely different coordinate systems on the same two-dimensional parameter space. The map
\[
  (m,p)\mapsto (\alpha,\beta)
\]
is one-to-one on $\{m>0,p>0,m\neq p\}$ with inverse
\[
  m = \alpha\beta,
  \qquad
  p = \alpha(2-\beta),
\]
so no new degrees of freedom are introduced by trading $(m,p)$ for $(\alpha,\beta)$. In logarithmic coordinates $M:=\log m$, $P:=\log p$, the EPIC index can be written as
\[
  \chi = \frac{1}{1+e^{P-M}},
\]
so that level sets of $\chi$ are straight rays through the origin in the $(M,P)$ plane; in this sense $\chi$ plays the role of a polar angle in log-space. The interest of the dictionary lies in showing which combinations of observables must coincide whenever a branched network has flowed, under coarse-graining, to such a ledger in a given scale window.

For later use it is convenient to summarise the parameter regimes. For $0<m<p$ one has $0<\beta<1$ and $0<\chi<\tfrac12$, so the flux-only cost \eqref{eq:flux-only-Gilbert} is strictly concave and the symmetric Y-angle \eqref{eq:theta-mp} is real and nondegenerate ($0<\theta<120^\circ$). The boundary case $m=p$ gives $\beta=1$, $\chi=\tfrac12$, and $\cos(\theta/2)=1$: the effective cost is linear in $|Q|$, there is no strict energetic incentive for trunk sharing, and the symmetric Y-junction degenerates to colinear branches. For $m>p$ one has $\beta>1$ and the right-hand side of \eqref{eq:theta-mp} exceeds $1$, signaling that the simple symmetric Y-geometry is incompatible with the homogeneous ledger. Accordingly, geometric interpretations of $\theta$ are restricted to $0<m\le p$, and Gilbert-type concavity arguments are restricted to $0<m<p$.
\end{remark}

\section{Homogeneity II: homogeneous normal form and rigidity}
\label{sec:rigidity}

The preceding sections started from the specific two-term ledger~\eqref{eq:ledger-def} and, by examining \emph{local} variations of the network functional, deduced flux--radius relations, flux-only concavity, Young--Herring-type angle balances, and the single-index Snell--Murray dictionary. In this ``forward'' direction, the two-monomial ledger is taken as given and the geometric and energetic consequences collected in Lemmas~\ref{lem:flux-radius}--\ref{lem:beta} and Theorem~\ref{thm:chi-dictionary} follow as necessary stationarity conditions. The present section addresses the complementary, ``reverse'' direction: within the general quadratic ledger class~\eqref{eq:quadratic-decomposition}, the central question is under what conditions simultaneous power-law behavior of the branchwise minimizer, the minimized flux-only cost, and the Young--Herring node balance forces the ledger back into the two-term form.

\subsection{Universality within the two-term ledger class}

Before turning to rigidity, it is helpful to collect the forward results of Secs.~\ref{sec:murray}--\ref{sec:chi-dictionary} in a single place. In a scale-free regime governed by the two-term ledger~\eqref{eq:ledger-def}, any configuration that is locally stationary with respect to edge radii and interior node translations necessarily satisfies a small set of linked structures:

\begin{enumerate}[label=(\roman*),nosep]
  \item \emph{Branchwise optimality.} Minimizing the ledger with respect to $r$ at fixed $Q$ yields a flux--radius power law $|Q|\propto r^\alpha$ with $\alpha=(m+p)/2$. At single-parent, two-daughter degree-3 nodes this implies the generalized Murray closure $r_0^\alpha = r_1^\alpha + r_2^\alpha$ (Lemma~\ref{lem:flux-radius} and Corollary~\ref{cor:murray}).

  \item \emph{Node translations.} Stationarity under translations of interior nodes, with $\{Q_e,r_e\}$ held at their branchwise optima, enforces a Young--Herring-type vector balance $\sum_{e\in E(v)} r_e^m \hat{\bm t}_e = \bm 0$ and fixes the symmetric Y-junction opening angle via $\cos(\theta/2)=2^{(m-p)/(m+p)}$ (Sec.~\ref{sec:angles}).

  \item \emph{Flux-only concavity.} Eliminating the radii using the branchwise optimum gives an effective flux-only edge cost $\mathcal P_{\rm br}^\ast(Q)\propto |Q|^\beta$ with $\beta=2m/(m+p)$, and an equivalent flux-only node balance $\sum_{e\in E(v)} |Q_e|^\beta \hat{\bm t}_e = \bm 0$ (Sec.~\ref{sec:concavity}).

  \item \emph{Single-index dictionary.} The exponents $\alpha$ and $\beta$, the symmetric Y-angle $\theta$, and the ledger exponents $(m,p)$ are all encoded in a single index $\chi=m/(m+p)$, obeying the Snell--Murray dictionary of Theorem~\ref{thm:chi-dictionary}.
\end{enumerate}

In other words, within the two-term ledger class the objects
\[
  (\alpha,\beta,\theta,\chi)
\]
are different faces of the same underlying homogeneous structure. The remainder of this section addresses the complementary question: within the broader quadratic ledger class~\eqref{eq:quadratic-decomposition}, under what homogeneity assumptions on the branchwise minimizer $r^\ast(Q)$ and the minimized cost $\mathcal P^\ast(Q)$ is one \emph{forced} back into this two-term form?

\subsection{Rigidity in the quadratic ledger class}
The forward direction describes what follows \emph{given} a two-term ledger. The complementary question is whether, within the broader quadratic ledger class~\eqref{eq:quadratic-decomposition}, the simultaneous appearance of power-law flux--radius relations and power-law flux-only costs forces the ledger back into that class. The following rigidity theorem addresses this question under strong homogeneity assumptions on the branchwise minimizer $r^\ast(Q)$ and the minimized cost $\mathcal P^\ast(Q)$. Once these assumptions are satisfied and the ledger is shown to be of two-term form, the radius-weighted and flux-only Young--Herring balances then follow automatically from the forward analysis.

\begin{theorem}[Rigidity in the quadratic ledger class under strong homogeneity]
\label{prop:rigidity-quadratic}
Let
\[
  \mathcal P(Q,r) = A(r)\,Q^2 + B(r),
\]
with $A,B\in C^2((0,\infty))$ strictly positive. Suppose there exists a nonempty open cone of fluxes $C\subset(0,\infty)$ such that:
\begin{enumerate}[label=(\roman*),nosep]
  \item For each $Q\in C$ there is a unique minimizer $r^\ast(Q)>0$ of $\mathcal P(Q,r)$ over $r>0$, and the map $Q\mapsto r^\ast(Q)$ is \emph{locally} homogeneous of degree $1/\alpha$ in the sense that, for each $Q\in C$, there exists $\varepsilon(Q)>0$ with
        \[
          r^\ast(\lambda Q) = \lambda^{1/\alpha} r^\ast(Q)
        \]
        for all $\lambda\in(1-\varepsilon(Q),1+\varepsilon(Q))$ such that $\lambda Q\in C$, for some exponent $\alpha>0$.
  \item The minimized per-length cost $\mathcal P^\ast(Q):=\mathcal P(Q,r^\ast(Q))$ is \emph{locally} homogeneous of degree $\beta$ in $Q$ on $C$ in the analogous sense: for each $Q\in C$ there exists $\delta(Q)>0$ with
        \[
          \mathcal P^\ast(\lambda Q) = \lambda^\beta \mathcal P^\ast(Q)
        \]
        for all $\lambda\in(1-\delta(Q),1+\delta(Q))$ such that $\lambda Q\in C$, for some exponent $0<\beta<2$.

\end{enumerate}
Let $I\subset(0,\infty)$ be a connected component of the image $r^\ast(C)\subset(0,\infty)$, and assume in addition that
\[
  A'(r)\neq 0 \quad\text{for all } r\in I,
\]
so that $-B'(r)/A'(r)$ is well defined on $I$. Then, on the interval $I$, there exist exponents $m,p>0$ and positive constants $a,b$ such that
\begin{equation}
  \label{eq:quadratic-monomial-conclusion}
  A(r) = a\,r^{-p},
  \qquad
  B(r) = b\,r^{m},
\end{equation}
and hence
\[
  \mathcal P(Q,r) = a\,Q^2 r^{-p} + b\,r^m,
\]
with
\[
  m = \alpha\beta,
  \qquad
  p = \alpha(2-\beta),
  \qquad
  m+p = 2\alpha.
\]
\end{theorem}

The proof proceeds in four steps. First, branchwise optimality is encoded as a relation $Q^2=R(r)$ between flux and radius at the minimizer. Second, the homogeneity of $r^\ast(Q)$ is used to show that $R(r)$ is a monomial in $r$ on the chosen connected component $I$ of its image. Third, the minimized ledger is recast as a second combination $F(r)$ of $A$ and $B$, and homogeneity of $\mathcal P^\ast(Q)$ in $Q$ is translated into homogeneity of $F$ in $r$ on $I$. Finally, the resulting functional equations for $R$ and $F$ are solved to show that $A(r)$ and $B(r)$ must themselves be monomials on $I$, recovering the two-term ledger form.

\begin{proof}
The first-order optimality condition for minimizing $\mathcal P(Q,r)$ with respect to $r$ at fixed $Q\in C$ reads
\[
  \partial_r \mathcal P(Q,r^\ast(Q))
  =
  A'(r^\ast(Q))\,Q^2 + B'(r^\ast(Q)) = 0.
\]
By the nonvanishing assumption on $A'$ this can be solved for $Q^2$ as a function of $r$ on $I$, giving
\begin{equation}
  Q^2 = -\,\frac{B'(r)}{A'(r)} =: R(r)
  \label{eq:R-def}
\end{equation}
for all $r\in I$ for which there exists $Q\in C$ with $r=r^\ast(Q)$. In particular, $R$ is well defined and positive on $I$.

By assumption~(i), $r^\ast(Q)$ is locally homogeneous of degree $1/\alpha$ on $C$. Fix $r\in I$ and choose $Q\in C$ with $r=r^\ast(Q)$; such a $Q$ exists by the definition of $I$ as a connected component of $r^\ast(C)$. There exists $\varepsilon>0$ such that
\[
  r^\ast(\lambda Q) = \lambda^{1/\alpha} r^\ast(Q)
\]
for all $\lambda\in(1-\varepsilon,1+\varepsilon)$ with $\lambda Q\in C$. For $\lambda$ sufficiently close to $1$, the points $\lambda^{1/\alpha}r$ remain in the same connected component $I$, so $R(\lambda^{1/\alpha}r)$ is defined there. Applying~\eqref{eq:R-def} to $(\lambda Q,r^\ast(\lambda Q))$ yields, for such $\lambda$,
\[
  R\big(\lambda^{1/\alpha} r\big)
  =
  (\lambda Q)^2
  =
  \lambda^2 R(r).
\]
Define
\[
  g(\lambda) := R\big(\lambda^{1/\alpha} r\big) - \lambda^2 R(r).
\]
The function $g$ is differentiable and vanishes for all $\lambda$ in a neighborhood of $1$, so in particular $g'(1)=0$. Differentiating at $\lambda=1$ gives
\[
  0 = g'(1)
  =
  \frac{1}{\alpha} r R'(r) - 2 R(r),
\]
that is,
\begin{equation}
  r R'(r) = 2\alpha R(r)
\end{equation}
for all $r\in I$. On the connected interval $I$ the general solution of this Euler ODE is
\begin{equation}
  R(r) = c\,r^{2\alpha}
  \label{eq:R-hom}
\end{equation}
for some constant $c>0$. Combining~\eqref{eq:R-def} and~\eqref{eq:R-hom} gives
\begin{equation}
  -\,\frac{B'(r)}{A'(r)} = c\,r^{2\alpha}
  \quad\Longrightarrow\quad
  B'(r) = -c\,r^{2\alpha} A'(r)
  \label{eq:Bprime-Aprime-relation}
\end{equation}
for all $r\in I$.

Next define
\begin{equation}
  F(r) := \mathcal P^\ast(Q)
  \quad\text{with } r = r^\ast(Q).
\end{equation}
Equivalently, using~\eqref{eq:R-def},
\begin{equation}
  F(r)
  =
  A(r)\,R(r) + B(r)
  =
  c\,r^{2\alpha} A(r) + B(r).
  \label{eq:F-def}
\end{equation}
Fix again $r\in I$ and write $r=r^\ast(Q)$ for some $Q\in C$. Assumption~(ii) states that $\mathcal P^\ast(Q)$ is locally homogeneous of degree $\beta$ in $Q$ on $C$, so there exists $\delta>0$ such that
\[
  \mathcal P^\ast(\lambda Q) = \lambda^\beta \mathcal P^\ast(Q)
\]
for all $\lambda\in(1-\delta,1+\delta)$ with $\lambda Q\in C$. Using the local homogeneity of $r^\ast$ from (i), this is equivalent, for $\lambda$ sufficiently close to $1$ such that $\lambda^{1/\alpha}r\in I$, to
\[
  F\big(\lambda^{1/\alpha} r\big)
  =
  \lambda^\beta F(r).
\]
Define
\[
  h(\lambda) := F\big(\lambda^{1/\alpha} r\big) - \lambda^\beta F(r).
\]
Then $h$ is differentiable and vanishes on a neighborhood of $1$, so $h'(1)=0$. Differentiating at $\lambda=1$ yields
\[
  0 = h'(1)
  =
  \frac{1}{\alpha} r F'(r) - \beta F(r),
\]
hence
\begin{equation}
  r F'(r) = \alpha\beta\,F(r)
\end{equation}
for all $r\in I$. Writing $\gamma:=\alpha\beta$, the general solution of this Euler ODE on $I$ is
\begin{equation}
  F(r) = k\,r^\gamma
  \label{eq:F-hom}
\end{equation}
for some constant $k>0$.

Differentiating~\eqref{eq:F-def} with respect to $r$ and using~\eqref{eq:Bprime-Aprime-relation} gives
\[
  F'(r)
  =
  c\,\big(2\alpha r^{2\alpha-1} A(r) + r^{2\alpha}A'(r)\big)
  +
  B'(r)
  =
  c\,\big(2\alpha r^{2\alpha-1} A(r) + r^{2\alpha}A'(r)\big)
  -c\,r^{2\alpha}A'(r),
\]
so the $A'(r)$ terms cancel and
\begin{equation}
  F'(r) = 2\alpha\,c\,r^{2\alpha-1} A(r).
  \label{eq:Fprime-from-A}
\end{equation}
On the other hand, differentiating the homogeneous form~\eqref{eq:F-hom} yields
\begin{equation}
  F'(r) = k\,\gamma\,r^{\gamma-1}.
  \label{eq:Fprime-from-hom}
\end{equation}
Equating~\eqref{eq:Fprime-from-A} and~\eqref{eq:Fprime-from-hom} results in
\[
  2\alpha\,c\,r^{2\alpha-1} A(r)
  =
  k\,\gamma\,r^{\gamma-1},
\]
so
\begin{equation}
  A(r)
  =
  \frac{k\,\gamma}{2\alpha c}\,r^{\gamma-2\alpha}
  =
  a\,r^{\gamma-2\alpha}
\end{equation}
for some constant $a>0$. Substituting this into~\eqref{eq:F-def} and using~\eqref{eq:F-hom} yields
\[
  k\,r^\gamma
  =
  F(r)
  =
  c\,r^{2\alpha} A(r) + B(r)
  =
  c\,r^{2\alpha} a\,r^{\gamma-2\alpha} + B(r)
  =
  c\,a\,r^\gamma + B(r),
\]
so
\begin{equation}
  B(r) = (k - c a)\,r^\gamma = b\,r^\gamma
\end{equation}
for some constant $b>0$. Hence both $A$ and $B$ are monomials on $I$:
\[
  A(r) = a\,r^{\gamma-2\alpha},
  \qquad
  B(r) = b\,r^\gamma.
\]

Finally, set $m:=\gamma=\alpha\beta$ and $p:=2\alpha-\gamma=\alpha(2-\beta)$. Then
\[
  A(r) = a\,r^{-p},
  \qquad
  B(r) = b\,r^m,
\]
with $m+p=2\alpha$, which recovers the two-term ledger form~\eqref{eq:ledger-def} on $I$.
\end{proof}

\begin{remark}
Theorem~\ref{prop:rigidity-quadratic} classifies ledgers within the quadratic class~\eqref{eq:quadratic-decomposition} that exhibit exact scale-free behavior of the primary observables, under the regularity assumptions built into its statement (existence, for $Q$ in an open flux cone, of a unique interior branchwise minimizer $r^\ast(Q)$, local homogeneity of $r^\ast$ and $\mathcal P^\ast$ in $Q$, and nonvanishing $A'$ on the relevant radius interval $I$). Under these hypotheses, if $\mathcal P(Q,r)=A(r)Q^2+B(r)$ admits a branchwise minimizer $r^\ast(Q)$ and minimized cost $\mathcal P^\ast(Q)$ that are \emph{exact} homogeneous power laws in $Q$ on some open scaling cone, then $A$ and $B$ must themselves be pure monomials in $r$ on each connected component of the image of $r^\ast$, and the ledger reduces, in that regime, to the two-term EPIC form~\eqref{eq:ledger-def}. Conceptually, this is the converse of Lemma~\ref{lem:scale-free}: there the homogeneity ansatz in $(Q,r)$ forced the ledger into the two-term class, whereas here homogeneous behavior of the optimal branchwise quantities in $Q$ forces the same normal form. Many continuum growth or transport models (Stefan problems, phase-field models, landscape evolution equations, and related systems) are expected to exhibit only asymptotic power-law behavior, possibly decorated by slowly varying prefactors and curvature- or anisotropy-dependent corrections. Such corrections are excluded by the strong homogeneity and non-degeneracy assumptions of the theorem. In those settings, the two-term ledger should be viewed as a leading homogeneous normal form for a scale window where sub-leading effects are small, rather than as an exact microscopic ledger; in this precise sense the theorem shows uniqueness of the two-term form within the quadratic ledger class under the stated homogeneity, interior-minimizer, and $A'\neq 0$ hypotheses.
\end{remark}

\subsection{A concrete diffusive example}

A simple diffusive example illustrates the rigidity theorem with a fully explicit model. Linear diffusion in a straight cylinder with volume-based upkeep already fits into the hypotheses and conclusions above.

\begin{lemma}[Linear diffusion with volume upkeep as an EPIC example]
\label{lem:diffusion-example}
Fix a spatial dimension $d\ge 2$ and consider a straight cylindrical branch of length $\ell$ and radius $r>0$ with cross-sectional area $A(r)=c_d r^{d-1}$ for some geometric constant $c_d>0$. In line with Sec.~\ref{subsec:EPIC-principle}, the transport exponent $p$ is defined by the cross-section scaling $A(r)\propto r^{p}$; for a circular cylinder in $d=3$ this gives $p=2$, while thin films or slabs would have different effective $p$. Let a scalar potential $\phi$ (temperature, concentration, or chemical potential) satisfy a stationary linear diffusion equation with constant diffusivity, and let $Q$ denote the total diffusive flux through each cross-section. Let the structural upkeep cost per unit length scale with volume as
\[
  \mathcal P_{\rm struct}(r) = b_0\,r^{m},
  \qquad
  m:=d-1,
  \quad
  b_0>0.
\]

Then, as in Sec.~\ref{subsec:EPIC-principle}:
\begin{enumerate}[label=(\roman*),nosep]
  \item The diffusive entropy-production rate per unit length is
        \[
          T_{\rm env}\,\frac{\dot S_{\rm prod}}{\ell}
          =
          a_0\,\frac{Q^2}{r^{d-1}},
        \]
        with $a_0>0$ depending only on $T_{\rm env}$, the diffusion coefficient, and geometric constants. The per-length ledger
        \[
          \mathcal P(Q,r)
          :=
          T_{\rm env}\,\frac{\dot S_{\rm prod}}{\ell}
          + \mathcal P_{\rm struct}(r)
          =
          a_0\,Q^2 r^{-(d-1)} + b_0\,r^{d-1}
        \]
        is therefore of the quadratic form~\eqref{eq:quadratic-decomposition} with
        \[
          A(r) = a_0\,r^{-(d-1)},
          \qquad
          B(r) = b_0\,r^{d-1}.
        \]
  \item For each $Q\neq 0$ the branchwise minimizer $r^\ast(Q)$ exists, is unique, and scales as a power of $Q$, with
        \[
          \alpha = \frac{m+p}{2}
          = \frac{(d-1)+(d-1)}{2}
          = d-1,
          \qquad
          p = d-1.
        \]
  \item The minimized per-length cost $\mathcal P^\ast(Q):=\mathcal P(Q,r^\ast(Q))$ is homogeneous of degree
        \[
          \beta = \frac{2m}{m+p} = \frac{2(d-1)}{2(d-1)} = 1
        \]
        in $Q$, i.e.
        \[
          \mathcal P^\ast(\lambda Q) = \lambda\,\mathcal P^\ast(Q)
          \quad\text{for all }\lambda>0.
        \]
\end{enumerate}
In particular, assumptions~(i) and~(ii) of Theorem~\ref{prop:rigidity-quadratic} are satisfied on any interval $I\subset(0,\infty)$ of radii, with
\[
  m = d-1,
  \qquad
  p = d-1,
  \qquad
  \alpha = d-1,
  \qquad
  \beta = 1,
\]
and the conclusion reduces to the explicit monomial forms of $A$ and $B$ given above.
\end{lemma}

This diffusive example sits exactly on the boundary $\beta=1$ between strictly concave and linear flux-only costs. In this case the EPIC ledger still belongs to the two-term class, but the effective cost is merely additive rather than strictly subadditive in $Q$: combining fluxes neither lowers nor raises the cost to leading order. As a result there is no energetic preference for trunk sharing, and the symmetric Y-junction angle degenerates to $\theta=0$ in the angle law~\eqref{eq:theta-mp}. More complex PDE-based models that incorporate anisotropy, curvature effects, or nonlinear kinetics are expected to approach this structure only asymptotically. In such settings, within the quadratic ledger class defined above, the rigidity theorem singles out the two-term ledger as the unique quadratic fixed point compatible with exact scale-free behavior of the primary network observables. It is therefore best viewed as a borderline test case for Theorem~\ref{prop:rigidity-quadratic}, rather than as a representative point in the strictly concave, trunk-sharing EPIC regime discussed elsewhere. More generally, Theorem~\ref{prop:rigidity-quadratic} shows uniqueness of the two-term form within the class of quadratic ledgers whose optimal responses $r^\ast(Q)$ and $\mathcal P^\ast(Q)$ are exact power laws in $Q$ on an open scaling cone, with interior minimizers and nonvanishing $A'$ on the image of $r^\ast$.

\section{Examples and diagnostics}
\label{sec:applications}

The abstract EPIC ledger~\eqref{eq:ledger-def} and index $\chi$ provide a compact parametrization of a broad class of scale-free branched networks. This section illustrates that structure in three concrete settings and shows how existing measurements constrain the effective index $\chi$ and suggest specific tests. In each case, the transport exponent $p$ is fixed by elementary physics, the structural exponent $m$ encodes an upkeep scaling, and the associated quadruple $(\alpha,\beta,\theta,\chi)$ follows from the dictionary of Secs.~\ref{sec:murray}--\ref{sec:chi-dictionary}.

For convenience,
\begin{equation}
  \alpha = \frac{m+p}{2},
  \qquad
  \beta = \frac{2m}{m+p},
  \qquad
  \chi = \frac{m}{m+p},
  \qquad
  \cos\frac{\theta}{2} = 2^{\beta-1},
  \label{eq:applications-dictionary}
\end{equation}
so that once any two of $\{\alpha,\beta,\theta\}$ and $p$ are known, the others are fixed.

\subsection{Poiseuille-type trees: vascular and microfluidic}
\label{subsec:poiseuille-toy}

In laminar Poiseuille flow through a cylindrical branch of radius $r$ and length $\ell$,
\[
  Q = \frac{\pi r^{4}}{8\mu \ell}\,\Delta P
\]
implies that, at fixed $Q$, the viscous dissipation per unit length scales as $Q^{2} r^{-4}$ \cite{BruusMicrofluidics,KirbyMicroNano}. The transport sector of the ledger therefore has exponent $p=4$. A minimal structural model takes the upkeep power per unit length to scale either with cross-sectional area (volume) or with perimeter (wall area):
\begin{itemize}[nosep]
  \item \emph{Volume-priced upkeep:} $\mathcal P_{\rm struct}\propto r^{2}$, i.e. $m=2$;
  \item \emph{Surface-priced upkeep:} $\mathcal P_{\rm struct}\propto r$, i.e. $m=1$.
\end{itemize}

\paragraph{Volume-priced upkeep.}

For $(p,m)=(4,2)$, the dictionary~\eqref{eq:applications-dictionary} gives
\[
  \alpha = 3,
  \qquad
  \beta = \frac{2}{3},
  \qquad
  \chi = \frac{1}{3},
  \qquad
  \cos\frac{\theta}{2} = 2^{-1/3}
  \ \Rightarrow\
  \theta \approx 75^\circ.
\]
An ideal Poiseuille network with volume-priced upkeep is therefore characterised by
\begin{itemize}[nosep]
  \item Murray exponent $\alpha = 3$ and closure $r_0^3 = r_1^3 + r_2^3$;
  \item concave flux-only cost $\sum_e \ell_e |Q_e|^{2/3}$;
  \item symmetric Y-junction angles near $75^\circ$;
  \item EPIC index $\chi = 1/3$, so that one third of the optimal ledger is transport and two thirds structural (Sec.~\ref{sec:chi-dictionary}).
\end{itemize}
This is exactly the Poiseuille+EPIC choice used in the retinal analysis of Ref.~\cite{BennettEPIC} and is broadly compatible with Murray-type exponents and angles in microvascular data \cite{Murray1926,ZamirOptimalityTrees,PriesSecombReview}.

Reported Murray exponents for blood vessels typically lie between $2$ and $3$, depending on vessel class and species \cite{Murray1926,ZamirOptimalityTrees,PriesSecombReview}. Large elastic arteries often exhibit $\alpha$ in the $2.3$--$2.6$ range, whereas smaller arterioles and venules, where Poiseuille flow is a better approximation, can be consistent with $\alpha$ closer to~3. Planar bifurcation angles in retinal or cortical microvasculature are broadly distributed, but sibling and daughter--daughter angles commonly have means in the $75^\circ$--$90^\circ$ band, with standard deviations of order $10^\circ$--$20^\circ$.

Combining these measurements with the EPIC relations yields a diagnostic for the existence of a Poiseuille+volume ledger.
\begin{itemize}[nosep]
  \item In a \emph{small-vessel Poiseuille window} with $\alpha \approx 2.6$--$3.0$ and $\theta \approx 80^\circ$--$90^\circ$, the dictionary produces two estimates of the concavity exponent,
        \[
          \beta^{(\alpha)} = 2-\frac{4}{\alpha},
          \qquad
          \beta^{(\theta)} = 1+\log_2\cos\frac{\theta}{2}.
        \]
    For example, $\alpha = 2.7$ and $\theta = 82^\circ$ yield $\beta^{(\alpha)}\approx 0.52$ and $\beta^{(\theta)}\approx 0.59$; the difference $|\beta^{(\alpha)}-\beta^{(\theta)}|\approx 0.08$ is comparable to the experimental scatter in $\alpha$ and $\theta$. Within uncertainties, a single local index $\chi \simeq \beta/2 \approx 0.28$ can plausibly describe such windows.
  \item In contrast, applying one Poiseuille+volume ledger to an \emph{entire vascular tree}, with a representative large-vessel exponent $\alpha\approx 2.4$ and $\theta\approx 90^\circ$, gives $\beta^{(\alpha)}\approx 1/3$ and $\beta^{(\theta)}\approx 1/2$, a discrepancy of order $0.17$ in $\beta$. Large vessels, with pulse dynamics, wall elasticity, curvature, and non-Newtonian effects, therefore lie outside any single quadratic, scale-free Poiseuille ledger.
\end{itemize}

\noindent
\textbf{Implication.}
The Poiseuille EPIC ledger is best interpreted as a \emph{local} effective description for microvascular or microfluidic tiles that are
\begin{itemize}[nosep]
  \item slender and approximately straight,
  \item laminar and approximately Poiseuille,
  \item and roughly scale-free over a few generations.
\end{itemize}
In such windows, existing geometric data are consistent with a single EPIC index $\chi$ within experimental uncertainty. Systematic surveys of $\alpha$ and $\theta$ versus scale in retinal or cortical trees would provide a direct test: in any scale window where a single ledger is appropriate, $\chi$ inferred independently from radii, angles, and, where accessible, energetic concavity should agree within errors.

\medskip\noindent
\textbf{Local energetic budget in a microvascular Y-junction.}

For a representative small-vessel window,
\[
  \alpha_{\rm rep} \approx 2.7,
  \qquad
  \theta_{\rm rep} \approx 82^\circ,
\]
the dictionary with $p=4$ gives
\[
  \chi^{(\alpha)} = 1 - \frac{p}{2\alpha_{\rm rep}}
  \approx 0.26,
  \qquad
  \chi^{(\theta)} = \frac{1}{2}\Bigl(1+\log_2\cos\tfrac{\theta_{\rm rep}}{2}\Bigr)
  \approx 0.30.
\]
These estimates differ by $\sim 0.04$, so an effective local index
\[
  \chi_{\rm micro} \approx 0.28\pm 0.04,
  \qquad
  \beta_{\rm micro} = 2\chi_{\rm micro} \approx 0.56\pm 0.08
\]
is consistent with the observed scatter.

Lemma~\ref{lem:chi-fraction} shows that the same index $\chi$ controls the ledger-sector split at the branchwise optimum,
\[
  \frac{\mathcal P_{\rm trans}}{\mathcal P_{\rm br}^\ast} = \chi,
  \qquad
  \frac{\mathcal P_{\rm struct}}{\mathcal P_{\rm br}^\ast} = 1-\chi.
\]
For a microvascular Y-tile in the window above, this implies
\begin{equation}
  \label{eq:EPIC-budget-prediction}
  \frac{\mathcal P_{\rm trans}}{\mathcal P_{\rm br}^\ast}
    \approx 0.25\text{--}0.30,
  \qquad
  \frac{\mathcal P_{\rm struct}}{\mathcal P_{\rm br}^\ast}
    \approx 0.70\text{--}0.75,
\end{equation}
that is, roughly $25$–$30\%$ of the local maintenance ledger is spent on viscous pumping and $70$–$75\%$ on structural upkeep. The effective flux-only cost scales as
\[
  \mathcal P_{\rm br}^\ast(Q) \propto |Q|^{\beta_{\rm micro}},
  \qquad
  \beta_{\rm micro} \approx 0.6.
\]

Equation~\eqref{eq:EPIC-budget-prediction} summarizes a geometry-based prediction for a quantity that is usually estimated indirectly: the local split between pumping and structural expenditure. Classical Murray-type arguments fix $\alpha$ by trading dissipation against volume or surface, but do not predict a concavity exponent for a flux-only functional or a numerical pumping fraction. Gilbert/OCN-style theories postulate a concavity exponent $\beta$ in a cost $\sum_e \ell_e |Q_e|^\beta$, but treat $\beta$ as an independent parameter. The EPIC framework links the observables $(\alpha,\theta,\beta)$ and the budget split through the single index $\chi$.

A direct test is conceptually straightforward in organ-on-chip systems or in high-resolution microvascular imaging:
\begin{itemize}[nosep]
  \item image a microvascular bed or microfluidic tree in a regime where $\alpha$ and
        $\theta$ can be robustly estimated;
  \item infer $\chi$ and $\beta$ from geometry via~\eqref{eq:applications-dictionary};
  \item independently measure how total pumping power or pressure drop scales with throughput
        (to estimate $\beta$), and/or estimate
        \[
          \frac{\text{viscous dissipation per unit length}}
               {\text{total per-length maintenance cost}};
        \]
  \item compare these measurements with the prediction in~\eqref{eq:EPIC-budget-prediction}.
\end{itemize}

\paragraph{Surface-priced upkeep.}

For completeness, the surface-priced case $(p,m)=(4,1)$ gives
\[
  \alpha = \frac{5}{2},
  \qquad
  \beta = \frac{2}{5},
  \qquad
  \chi = \frac{1}{5},
  \qquad
  \cos\frac{\theta}{2} = 2^{-3/5}
  \ \Rightarrow\
  \theta\approx 97^\circ.
\]
Here the ledger penalizes large radii more strongly, the flux-only cost is less concave, branches open out, and only about one fifth of the optimal ledger is transport. Microfluidic networks etched in a uniform substrate, where footprint or wall area dominates cost, are natural candidates to probe this regime; EPIC predicts systematically larger opening angles and weaker trunk sharing than in volume-priced vascular trees.

\subsection{Diffusive and Ohmic conduction trees}
\label{subsec:conduction-toy}

For diffusive or Ohmic transport in three dimensions along a straight cylindrical branch, the cross-sectional area scales as $A(r)\propto r^{2}$ and the entropy production per unit length scales as $Q^2/A(r)\propto Q^2/r^2$, so the transport exponent is $p=2$ \cite{DeGrootMazur,BruusMicrofluidics}. As before, a simple structural model takes the upkeep per unit length to scale with either volume ($m=2$) or surface ($m=1$).

For $(p,m)=(2,2)$,
\[
  \alpha = 2,
  \qquad
  \beta = 1,
  \qquad
  \chi = \frac{1}{2},
  \qquad
  \cos\frac{\theta}{2} = 2^{\beta-1} = 1
  \ \Rightarrow\
  \theta = 0^\circ.
\]
This is a limiting case: the flux-only cost $\propto |Q|^\beta$ is linear in $|Q|$, there is no strict energetic incentive for trunk sharing at leading order, and the symmetric Y-junction degenerates. Systems near this parameter choice are therefore expected to be highly sensitive to additional physics beyond the two-term homogeneous ledger (anisotropy, boundary conditions, nonlocal constraints), which can shape the network without a strong bias toward deep trees.

For $(p,m)=(2,1)$,
\[
  \alpha = \frac{3}{2},
  \qquad
  \beta = \frac{2}{3},
  \qquad
  \chi = \frac{1}{3},
  \qquad
  \cos\frac{\theta}{2} = 2^{\beta-1} = 2^{-1/3}
  \ \Rightarrow\
  \theta\approx 75^\circ,
\]
so the ledger behaves qualitatively like the Poiseuille+volume case, but with a weaker dependence of flux on radius ($|Q|\propto r^{3/2}$) and a slightly different balance between transport and structure. Branched thermal-conduction structures in electronics or porous media, in which wall-area cost is roughly linear in radius and diffusion dominates transport, are natural candidates for this regime.

\medskip\noindent
\textbf{Implication.}
The conduction example highlights that EPIC does not always predict strongly concave flux-only costs: in some parameter ranges it reduces to a nearly linear effective ledger, and geometry is then dominated by other constraints. When experimental exponents are close to $(\alpha,\beta)=(2,1)$ for $p=2$, the two-term ledger should be viewed as a homogeneous baseline around which system-specific corrections are organized, rather than as a sharply predictive theory of topology.

\subsection{Geophysical drainage and effective exponents}
\label{subsec:geo-toy}

Geophysical drainage networks introduce additional complexity: flows may be turbulent, channels have irregular cross sections, and erosion and sediment transport introduce long memory effects \cite{RodriguezIturbeRinaldo,RinaldoMinimumEnergy}. Nonetheless, many reduced models treat the discharge $Q$ as a scalar flux along an effective channel of width $w$ and depth $h$, with dissipation per unit length scaling as $Q^2/(w\,h^2)$ or a similar form, and structural cost scaling with cross-sectional area or perimeter \cite{RodriguezIturbeRinaldo,BejanConstructalReview,BanavarEfficientNetworks}. Empirically, classical hydraulic geometry shows that the bankfull channel width---the width at the stage where flow just fills the channel without overtopping the banks---obeys
\[
  w \propto Q^a, \qquad a \approx 0.5\pm 0.1,
\]
across many rivers \cite{LeopoldMaddock,WohlWilcox}. A central value $a\approx 0.5$ corresponds to an effective exponent $\alpha_{\rm emp}=1/a\approx 2$ for flux versus a ``radius'' proxy. Planform measurements further show that both confluences and bifurcations in many river basins have branching angles clustered around $70^\circ$--$75^\circ$, with broad but robust peaks in humid, seepage-dominated regions \cite{Seybold2017Climate}.

Interpreting large alluvial rivers as an EPIC-type diffusive network with $p\simeq 2$ leads to a sharp test. Using~\eqref{eq:applications-dictionary}:
\begin{itemize}[nosep]
  \item from $\alpha_{\rm emp}\approx 2$ and $p=2$,
        \[
          \beta^{(\alpha)} = 2-\frac{2}{\alpha_{\rm emp}} \approx 1,
        \]
        so EPIC would predict an approximately linear flux-only cost and degenerate symmetric
        Y-angles;
  \item from $\theta_{\rm emp}\approx 72^\circ$,
        \[
          \beta^{(\theta)} = 1+\log_2\cos\frac{\theta_{\rm emp}}{2} \approx 0.69,
        \]
        corresponding to a strongly concave flux-only cost.
\end{itemize}
The discrepancy $|\beta^{(\alpha)}-\beta^{(\theta)}|\approx 0.3$ is far larger than the observational uncertainty in these exponents. Large, mixed-process river networks therefore lie outside the quadratic, scale-free EPIC class with a single index $\chi$: they are governed by turbulence, threshold sediment motion, bank strength, vegetation, and climate history, all of which violate the simple homogeneity assumptions built into
\eqref{eq:ledger-def}.

The same angular data can, however, be used to define a quantitative target for regimes where EPIC might apply. Assuming $p\approx 2$ and taking $\theta\approx 72^\circ$ as the primary observable, the dictionary implies
\[
  \beta \approx 0.69,
  \qquad
  \alpha = \frac{p}{2-\beta} \approx 1.53,
  \qquad
  a = \frac{1}{\alpha} \approx 0.65.
\]
In a genuinely quadratic, scale-free drainage regime with $p\approx 2$ and robust $72^\circ$ branching, EPIC would therefore predict that
\begin{itemize}[nosep]
  \item the width--discharge exponent should be closer to $a\approx 0.64$--$0.66$ than to
        $0.5$;
  \item the concavity exponent linking energy dissipation to discharge should be
        $\beta\approx 0.7$.
\end{itemize}
These values are not typical of large alluvial rivers \cite{LeopoldMaddock,WohlWilcox}, but they are plausible targets for \emph{narrower} classes of drainage: groundwater-sapping valleys, experimental erosion tanks, small deltas, or other systems where flow is more nearly laminar or diffusive and a scale-free window may exist.

\medskip\noindent
\textbf{Implication.}
In drainage, EPIC acts more as a \emph{filter} than as a direct model: it rules out the idea that a single quadratic, scale-free ledger describes all rivers, but it also supplies quantitative benchmarks (the triplet $(a,\beta,\theta)$ above) against which to test
specially selected drainage regimes. Observing those triplets in seepage-dominated networks would be strong evidence that such systems have flowed into an EPIC-like universality class; persistent discrepancies would indicate that even those regimes are controlled by additional scales or nonlinearities.

\medskip
Across these examples, the role of the EPIC index $\chi$ is twofold. First, it organizes a family of observables---flux--radius exponents, concavity, and junction angles---into a single scalar descriptor of a scale-free regime. Second, it provides falsifiable predictions: in any purported EPIC window, $\chi$ inferred independently from $(\alpha,\beta,\theta)$ must agree within uncertainties. Where this fails (as in whole-tree vasculature or large river basins), the failure is informative: it locates the boundary of the quadratic, scale-free cone and points directly to the physics that lie beyond it. Where it succeeds (as in carefully
selected Poiseuille microvascular or microfluidic tiles, or in future seepage experiments), it converts geometry into concrete energetic predictions such as Eq.~\eqref{eq:EPIC-budget-prediction} that are not fixed by more conventional optimality frameworks.

\section{Discussion and outlook}
\label{sec:discussion}
The analysis in this paper begins from a deliberately minimal description of a branched network: each edge carries a scalar flux $Q$, an effective radius $r$, and a per-length ledger $\mathcal P(Q,r)$ that aggregates the costs of transport and upkeep. Two structural assumptions are imposed: linear-response quadratic dependence on $Q$ at fixed geometry, and the existence of an intermediate, approximately scale-free regime in which $\mathcal P$ is homogeneous in $(Q,r)$ on an open scaling cone. In such a regime, an Euler-type argument reduces any admissible quadratic ledger to a two-term power law
\[
  \mathcal P(Q,r) = a\,Q^2 r^{-p} + b\,r^m,
\]
and all of the geometric structure follows from minimizing this expression at fixed flux and enforcing stationarity with respect to node translations.

A first outcome is a compact set of structural relations that gather several classical results into a single framework. The branchwise flux--radius relation $|Q|\propto r^\alpha$ and generalized Murray closure $r_0^\alpha = r_1^\alpha + r_2^\alpha$ arise directly from the competition between the transport term $Q^2 r^{-p}$ and the structural term $r^m$. Stationarity of the ledger under infinitesimal motion of a junction enforces a vector balance $\sum r^m \hat{\bm t} = \bm 0$ identical in form to a Young--Herring triple-junction condition with effective line tensions $\propto r^m$; for symmetric Y-junctions this yields the angle law $\cos(\theta/2)=2^{(m-p)/(m+p)}$, recovering the Steiner $120^\circ$ limit as $m\to 0$ and predicting a degenerating angle as $m\to p$. Eliminating the radii at the branchwise optimum produces a flux-only cost $\mathcal P^\ast(Q)\propto |Q|^\beta$ with $\beta=2m/(m+p)$; for $0<m<p$ this cost is strictly concave, and the corresponding node balance $\sum |Q|^\beta \hat{\bm t}=\bm 0$ reproduces the standard Gilbert/branched-transport angle laws. Within the quadratic, scale-free ledger class, Murray-type scaling, Young--Herring angles, and Gilbert concavity thus appear as mutually constrained consequences of the same two-term structure rather than as independent ingredients.

A second outcome is the identification of a single dimensionless index
\[
  \chi = \frac{m}{m+p},
\]
which packages the ledger exponents into a form that is directly comparable across systems. The Snell--Murray dictionary of Sec.~\ref{sec:chi-dictionary} shows that, once the transport exponent $p$ is known, $\chi$ can be inferred from any two among the flux--radius exponent $\alpha$, the concavity exponent $\beta$, and the symmetric opening angle $\theta$, via
\[
  \chi
  =
  1 - \frac{p}{2\alpha}
  =
  \frac{\beta}{2}
  =
  \frac{1}{2}\Bigl(1+\log_2\cos\tfrac{\theta}{2}\Bigr).
\]
In the abstract ledger, $\chi$ measures how the optimal per-length ledger is split between the transport and structural sectors. In the EPIC interpretation, it is exactly the fraction of the local entropy-per-information budget spent on driving flux rather than maintaining structure. In this sense, $\chi$ acts as a coarse-grained, thermodynamic order parameter for branched networks operating in a quadratic, scale-free regime.

Perhaps the most structurally informative result is the rigidity theorem of Sec.~\ref{sec:rigidity}. Within the general quadratic class $\mathcal P(Q,r)=A(r)Q^2+B(r)$, the combination of (i) a pure power-law flux--radius relation $|Q|\propto r^\alpha$ at the branchwise minimizer and (ii) a pure power-law flux-only cost $\mathcal P^\ast(Q)\propto |Q|^\beta$ on an open flux cone forces $A(r)$ and $B(r)$ themselves to be monomials in $r$. In other words, whenever a branched system empirically exhibits power-law flux--radius scaling, power-law flux-only scaling, and power-law radius weights in a Young--Herring balance over a finite scale window, its effective quadratic ledger in that window must belong to the two-term family analyzed here. The two-term EPIC ledger is therefore not merely a convenient ansatz: under the stated homogeneity assumptions it is the unique representative of the quadratic, scale-free class compatible with simultaneous Murray-type scaling, Gilbert concavity, and Young--Herring angles.

The schematic examples in Sec.~\ref{sec:applications} illustrate how the index $\chi$ can be used as a diagnostic. In laminar Poiseuille networks, $p=4$ is fixed by the Hagen--Poiseuille law, and the combination of flux--radius data and symmetric junction angles constrains $\chi$, and hence the concavity exponent $\beta$ and the transport/structure budget split; the Poiseuille EPIC analysis of retinal trees in Ref.~\cite{BennettEPIC} corresponds to $\chi\simeq 1/3$. In diffusive or Ohmic networks, geometry fixes $p\simeq 2$, and the same dictionary applies. In coarse-grained drainage models, effective exponents inferred from bankfull width--discharge relations, power–discharge scaling, and planform junction angles map to corresponding values of $\chi$, which then serve as compact descriptors of distinct climatic or lithological regimes and, importantly, indicate where a quadratic, scale-free description breaks down.

Several theoretical directions follow naturally. One concerns the relaxation of the quadratic (linear-response) assumption in the transport sector. Many geophysical and technological flows are strongly nonlinear, with leading-order dissipation that scales as $|Q|^{1+\delta}$ at fixed geometry, with $\delta\neq 1$. Extending the homogeneity analysis of Sec.~\ref{sec:murray} to such cases would replace the exponent $2$ by a general transport degree $\gamma$ and modify the dictionary accordingly, yielding relations of the form $\alpha=(m+\tilde p)/\gamma$ and $\beta=\gamma m/(m+\tilde p)$ for an effective geometry exponent $\tilde p$. The qualitative message—that a small number of homogeneity exponents control a larger set of observable scalings and angles, and that their ratios play the role of order parameters—is expected to survive, but the explicit form of the dictionary would change.

A second extension concerns anisotropy and nonlocality. The present treatment models each branch as isotropic and the ledger as a strictly local function of $(Q,r)$. In systems with strong orientation-dependent surface energies (for example, faceted dendritic crystals) or with nonlocal control costs (for example, neural or vascular networks with global perfusion or wiring constraints), additional tensorial or nonlocal contributions inevitably appear in the ledger. One natural line of work is to generalize the Young--Herring balance \eqref{eq:bit-tension} to include orientation-dependent weights and to examine how the index $\chi$ interacts with angular anisotropy. Another is to connect the scale-free ledger picture to explicit renormalization-group treatments of branched transport and to identify conditions under which the two-term ledger appears as an infrared fixed point of more complex microscopic models.

On the empirical side, the framework suggests a concrete program of falsifiable tests. In any candidate system, at least two of $\{\alpha,\beta,\theta\}$ together with a transport exponent $p$ can, in principle, be estimated. The index $\chi$ inferred from those measurements then fixes the third observable via the Snell--Murray dictionary. In microfluidic trees, the Murray exponent and symmetric junction angles can be extracted from images, while the concavity exponent $\beta$ can be estimated by varying the number of outlets or throughput and measuring how pumping power scales. In drainage networks, width--discharge relations and confluence angles can be obtained from topography, and effective concavity exponents inferred from the scaling of energy dissipation with basin size. Agreement with the EPIC dictionary would support the presence of an underlying quadratic, scale-free ledger over the probed scale range; systematic departures would identify regimes where linear response, scale-free homogeneity, or local additivity fail.

The connection to the original EPIC formulation provides a thermodynamic interpretation of the exponents. In the J/bit framework of Ref.~\cite{BennettEPIC}, the structural exponent $m$ and the transport exponent $p$ arise from calibrating how control-layer power scales with ``bits of structure'' and how entropy production scales with flux and radius. Crucially, the present work does not introduce any additional variational principle (such as maximum entropy production) to select $(m,p)$ \cite{DewarMaxEP,MartyushevSeleznev,KleidonLorenz}; instead, $(m,p)$ enter as effective empirical descriptors of a regime that has already been coarse-grained to a quadratic, scale-free ledger. Once those exponents are measured, many aspects of network geometry are fixed, and the single index $\chi$ simultaneously quantifies a budget split in physical units and organizes a family of geometric and energetic exponents. This dual role makes $\chi$ a promising candidate for a genuine nonequilibrium descriptor of branched systems: it links microscopic balance sheets (energy per bit, effective dimensionality of transport, upkeep scaling) to macroscopic observables (flux--radius scaling, concavity, junction geometry) in a form that can be confronted directly with data.

Whether this single scalar index will remain robust across the full diversity of natural and engineered trees remains an open question. What the present analysis establishes is that, whenever a branched network can be described by an admissible quadratic, scale-free ledger on some scale window, the geometric freedom is sharply constrained: a single ratio,
\[
  \chi = \frac{m}{m+p},
\]
determines how much of the local ledger is spent on moving material versus maintaining structure, and the main features of the resulting geometry—the flux--radius law, the concavity of the effective cost, and the junction angles—follow accordingly.

\bibliography{ref}

\end{document}